\documentclass[orivec,a4paper]{llncs}

\usepackage{clrscode}
\usepackage{stmaryrd}
\usepackage{enumitem}
\usepackage{amsmath}
\usepackage{amssymb}

\newcommand{\schema}[1]{\ensuremath{\mathbf{#1}}}
\newcommand{\smap}[1]{\ensuremath{{#1}}}
\newcommand{\query}[1]{\ensuremath{{\MakeUppercase{#1}}}}
\newcommand{\vect}[1]{\ensuremath{\bar{#1}}}
\newcommand{\sol}[3]{\ensuremath{\llbracket(#1,\smap{#2})\rrbracket_{\sf{#3}}}}
\newcommand{\eval}[3]{\ensuremath
{\mbox{\sc EVAL}^{\sf{#3}}(\query{#1},\smap{#2})}}
\newcommand{\cert}[4]{\ensuremath{\mbox
{\bf cert}^{\sf{#4}}(\query{#1},(#2,\smap{#3}))}}

\newcommand{\insta}[1]{\ensuremath{Inst(\schema{#1})}}
\newcommand{\defined}{\ensuremath{\mbox{:=}}}
\newcommand{\setdef}{\ensuremath{\mbox{ : }}}
\newcommand{\comma}{\ensuremath{\mbox{,}}}
\newcommand{\cq}{\ensuremath{{\sf CQ}}}
\newcommand{\ucq}{\ensuremath{{\sf UCQ}}}
\newcommand{\cqn}{\ensuremath{{\sf CQ}^{\neg}}}
\newcommand{\cqnp}{\ensuremath{{\sf CQ}^{\neg,1}}}
\newcommand{\ucqn}{\ensuremath{{\sf UCQ}^{\neg}}}
\newcommand{\ucqin}{\ensuremath{{\sf UCQ^{\neq}}}}
\newcommand{\ucqino}{\ensuremath{{\sf UCQ^{\neq,1}}}}
\newcommand{\fo}{\ensuremath{{\sf FO}}}

\newcommand{\repe}[1]{\ensuremath{Rep}(#1)}
\newcommand{\ad}[1]{\ensuremath{\#_{\sf D}(#1)}}
\newcommand{\adr}[2]{\ensuremath{\#_{\sf D}(#1,#2)}}

\newcommand{\ac}[1]{\ensuremath{\#_{\sf C}(#1)}}
\newcommand{\acr}[2]{\ensuremath{\#_{\sf C}(#1,#2)}}

\newcommand{\aegd}{{\em{aegd}}}
\newcommand{\aegds}{{\em{aegds}}}
\newcommand{\abd}{{\em{abd}}}
\newcommand{\abds}{{\em{abds}}}
\newcommand{\tgd}{{\em{tgd}}}
\newcommand{\tgds}{{\em{tgds}}}
\newcommand{\egd}{{\em{egd}}}
\newcommand{\egds}{{\em{egds}}}
\newcommand{\sttgd}{{\em{s-t tgd}}}
\newcommand{\sttgds}{{\em{s-t tgds}}}

\newcommand{\cons}{\ensuremath{{\sf Cons}}}
\newcommand{\vars}{\ensuremath{{\sf Vars}}}
\newcommand{\nulls}{\ensuremath{{\sf Nulls}}}
\newcommand{\nullso}{\ensuremath{{\sf Nulls^{o}}}}
\newcommand{\nullsc}{\ensuremath{{\sf Nulls^{c}}}}
\newcommand{\nullv}{\ensuremath{\bot}}
\newcommand{\nullo}{\ensuremath{\bot^{\sf o}}}
\newcommand{\nullc}{\ensuremath{\bot^{\sf c}}}

\newcommand{\letbe}{\ensuremath{\gets}}
\newcommand{\dom}[1]{\ensuremath{dom(#1)}}
\newcommand{\domc}[1]{\ensuremath{cons(#1)}}

\newcommand{\mgu}{\ensuremath{mgu}}
\newcommand{\false}{\ensuremath{\bf{false}}}
\newcommand{\true}{\ensuremath{\bf{true}}}
\newcommand{\np}{\ensuremath{\sf{NP}}}
\newcommand{\conp}{\ensuremath{\sf{coNP}}}
\newcommand{\ptime}{\ensuremath{\sf{P}}}
\newcommand{\ar}[1]{\ensuremath{arity(#1)}}
\newcommand{\aff}[1]{\ensuremath{{\sf aff}(#1)}}
\newcommand{\abdset}{\ensuremath{\Sigma^{\leftrightarrow}}}
\newcommand{\abdsetn}{\ensuremath{\Sigma^{\leftrightarrow}_{\diamond}}}
\newcommand{\abdsetabd}{\ensuremath{\Sigma^{\leftrightarrow}_{abd}}}
\newcommand{\abdsetaegd}{\ensuremath{\Sigma^{\leftrightarrow}_{aegd}}}
\newcommand{\abdsetleft}{\ensuremath{\Sigma^{\leftarrow}}}
\newcommand{\abdsetright}{\ensuremath{\Sigma^{\rightarrow}}}

\newcommand{\body}[1]{\ensuremath{body(#1)}}

\newcommand{\annt}[2]{\ensuremath{#1^{#2}}}

\newcommand{\annot}{\ensuremath{annot}}
\usepackage{url}
\urldef{\mailsa}\path|Adrian.Onet@morganstanley.com|

\begin{document}

\mainmatter  % start of an individual contribution

\title{Inference-based semantics in Data Exchange}

\titlerunning{Inference-based semantics in Data Exchange}

\author{Adrian Onet}
\authorrunning{Adrian Onet}

\institute{Morgan Stanley, Montreal, Canada\\
\mailsa
}
%\\
%\url{http://www.springer.com/lncs}}


\toctitle{Inference-based semantics in Data Exchange}
\tocauthor{Adrian Onet}
\maketitle

\begin{abstract}
Data Exchange is an old problem that was 
firstly studied from a theoretical point of view only in 2003.
Since then many approaches were considered when
it came to the language describing the relationship 
between the source and the target schema. These 
approaches focus on
what it makes a target instance a ``good'' solution for
data-exchange.
In this paper we propose the inference-based semantics 
 that solves many 
certain-answer anomalies existing in current 
data-exchange semantics.
To this we introduce a new mapping language 
between the source and the target schema based on
{\em annotated bidirectional dependencies} (\abd)
and, consequently define the semantics for this new language.
It is shown that the ABD-semantics can properly 
represent the inference-based semantics, for any source-to-target
mappings.
We discovered three dichotomy results under the 
new semantics for solution-existence, solution-check and
$\ucq$ evaluation problems. These results rely on 
 two factors describing the annotation used 
in the mappings (density and
cardinality).
Finally we also investigate 
the certain-answers evaluation problem under ABD-semantics
and discover many tractable classes for non-$\ucq$ queries
even for a subclass of~$\cqn$.
 \end{abstract}

\bigskip
\section{Introduction}\label{intro}
The data-exchange problem is that of transforming a database existing under 
a source schema
into another database under a different target schema.
This database transformation is based on mappings that describe the 
relationship between
the source and the target database. A mapping $M$ can be viewed as a, possibly 
infinite, set of
pairs $(I,J)$, where $I$ is a source instance and $J$ a target instance.
In this case, $J$ is called a {\em data-exchange solution} for $I$ and $M$.
The mapping between the source and the target database is usually specified in some
logic formalism.
The most widely accepted mapping language is the one based on sets of 
tuple-generating
dependencies (\tgds) and equality-generating dependencies (\egds).
A tuple-generating dependency is a $\fo$ sentence of the form:
$\forall \vect{x} (\forall \vect{y}\; \alpha(\vect{x},\vect{y}) \rightarrow 
\exists \vect{z}\; \beta(\vect{x},\vect{z}))$,
where $\alpha$ and $\beta$
represent conjunctions of atoms. 
In case all atoms from $\alpha$ are over the source schema and all the 
atoms from  $\beta$
are over the target schema, then we name the dependency a 
{\em source-to-target tgd} (\sttgd).
An equality-generating dependency (\egd) is a $\fo$ 
sentence of the form:
$\forall \vect{x} (\alpha(\vect{x}) \rightarrow x=y)$,
where $\alpha(\vect{x})$ is a conjunction of atoms over the target schema
and variables $x$ and $y$ occur in $\vect{x}$.

There may be several solutions for a given source instance
and a mapping, thus 
the set of solutions can be viewed  as
an incomplete database 
\cite{DBLP:journals/jacm/ImielinskiL84,DBLP:books/sp/Grahne91}.
Most commonly, the target solutions need to be 
queried. This takes us to the problem of querying incomplete databases 
\cite{DBLP:conf/sigmod/AbiteboulKG87}.
The {\em exact answers semantics} defines
the answer to a query over an incomplete database as the set of
all answers to the query for each instance from its semantics.
This approach is rarely feasible in practice, 
therefore two
approximations are commonly accepted: 
{\em certain answer} (answers occurring for all solutions)
and {\em maybe answer semantics} (answers occurring for some solutions).

Based on the interpretation of the mapping language there may be
several semantics in data exchange.
The first semantics was introduced in 
\cite{DBLP:journals/tcs/FaginKMP05} and it is here referred to as 
the OWA-semantics (open world assumption).
In order to improve 
the certain answer behavior for non-$\ucq$ queries under OWA-semantics,
Libkin \cite{DBLP:conf/pods/Libkin06}  introduced the first
closed-world semantics in data exchange,
known as CWA-semantics. This semantics was later on extended 
by Hernich and Schweikardt \cite{DBLP:conf/pods/HernichS07}
to include target \tgds.
In CWA-semantics a solution incorporates only tuples that are ``justified''
by the source and the dependencies.

Hernich \cite{DBLP:phd/de/Hernich2010} 
enumerated three rules that a ``good'' data-exchange semantics should follow:
1) Implicit information in schema mapping and source instances are
taken into account;
2) Logical equivalence of schema mapping is respected; and
3) The standard semantics of $\fo$ quantifier is reflected.
Based on these rules 
Hernich introduced the GCWA$^{*}$-semantics for data exchange.
Unfortunately most of these semantics have a rather
strange behavior when looking for certain answers over the set of solutions.
More recently Arenas et al. \cite{DBLP:conf/amw/ArenasD014}
introduced a new semantics for data-exchange based on bidirectional 
constraints. This new semantics solves most of the 
query anomalies present in the other semantics
but it comes with price of non-tractable
data complexity even for the simplest data-exchange problems.
Also, as we will see, there are simple mappings that can be specified
by a set of \tgds, under a closed-world semantics, 
but can not be specified by only using bidirectional 
constraints.

To better understand the type of
anomalies we may encounter 
under these semantics, 
consider the
next simple example. 
A company 
has in the source schema a binary relation $P$ 
representing the relationship between projects and 
employees.
After some reorganization it was realized that each projects 
financing should be provided by one or more cost-centers,
each employee belonging to one of these cost-centers.
With this the company creates the target schema
with two binary relations $PC$ and $CE$ representing the 
project to cost-center and cost-center to employee relationship respectively.
In the case of the unidirectional
semantics the process can be specified by the following \tgd: 
\begin{equation}\label{cwa-notworking1}
\begin{split}
\xi_1: \forall p\; \forall e\; P(p,e) \rightarrow \exists cc\; PC(p,cc)\wedge CE(cc,e).
\end{split}
\end{equation}

\noindent 
Let source $I$
be  $P^{I}=\{ (p_1,e_1), (p_1,e_2), (p_2,e_3) \}$,
stating that there are two employees $e_1$ and $e_2$ working on project 
$p_1$ and only employee $e_3$ working on project~$p_2$.
Consider target boolean query:
$\query{q}\defined \forall p,cc\; PC(p,cc)\wedge CE(cc,"e_3")\rightarrow p="p_2"$
({\em Is $e_3$ involved only in project $p_2$?}).
Because target instance $J$,
with $CE^{J}=\{ (cc_1,e_1), (cc_1,e_2), (cc_1,e_3) \}$
and with $PC^{J}=\{ (p_1,cc_1), (p_2,cc_1) \}$,
is part of all unidirectional semantics,
the certain answer for the given query will be
un-intuitively \false\;
under these semantics.
One may expect this answer to be \true\; based input data,
thus instance $J$ should not be part of the semantics.
For the same problem, consider the 
mapping represented by the next bidirectional constraint:
\begin{equation}\label{cwa-notworking2}
\begin{split}
\xi_2: \forall p\; \forall e\; \big( P(p,e) \leftrightarrow 
 \exists cc\; PC(p,cc)\wedge CE(cc,e)\big).
\end{split}
\end{equation}
\noindent
This mapping together with the bidirectional-constraints 
semantics
solves the previous anomaly but introduces
others. Consider query:
``{\em Does each cost-center have an employee?}''
Under bidirectional constraint semantics the 
certain answer to this query will be counter-intuitively 
\false. This happens because bidirectional semantics does not 
require target tuples to be ``justified'', thus 
 instance $J'$, with
 $CE^{J'}=\{ (cc_1,e_1), (cc_1,e_2), (cc_2,e_3) \}$
and $PC^{J'}=\{ (p_1,cc_1), (p_2,cc_2),(p_2,cc_3)  \}$ is part of 
the bidirectional constraint semantics. 
Note that this anomaly can be fixed with a rather 
complicated bidirectional constraint mapping where the right-hand 
side is a first-order expression.

{\bf Contributions}
Motivated by certain answer anomalies and 
expressivity issues in the current data-exchange semantics,
in this paper we propose a new data-exchange semantics
based on logical inference
for mappings represented by 
sets of  \sttgds\; and target \egds. This  semantics
eliminates most
 anomalies related to certain-answers and 
keeps the same certain answers with the other semantics
for union of conjunctive queries.
To this, we show that for any set of \sttgds\; and safe \egds\;
the inference-based semantic can easily be 
represented in a much richer language of 
annotated bidirectional dependencies (abd)
and safe annotated 
target \egds\footnote{Intuitively safe \egds\; do not 
allow joining with attributes that may contain null values.}
(safe \aegd).
The restriction to safe \aegds\; instead of regular \egds\;
exists because by
allowing non-safe \aegds, it 
opens the door to certain answer anomalies.
% (see Example \ref{safeaegd}).
%Note that this is not really a restriction compared with the
%other data-exchange semantics as 
%they all have
%an odd behavior in the presence of non-safe \egds\; 
%.
%With this restriction under ABD-semantics we will not be able 
%to express primary-key constraints over attributes that 
%may contain unknown values.

We discovered two important characteristics for each set $\abdset$
of \abds, 
namely the {\em annotation density} ($\ad{\abdset}$) and
{\em annotation cardinality} ($\ac{\abdset}$).
Intuitively, the annotation density measures the number of occurrences
a relational symbol is annotated with the same label in~$\abdset$.
The annotation cardinality measures the number
of labels used for a relational symbol in~$\abdset$.
We found that data-exchange 
solution-existence and solution-check problems
have a dichotomic behavior
based on annotation density and cardinality, respectively.

On the expressibility side, we show that
for any set $\Sigma$ of \sttgds\; and \egds\;
one can simply compute $\abdset$, a set of \abds\; and \egds\,
of density 1 such that for any instance $I$ 
the inference-based semantics for $\Sigma$ and $I$ coincide with the 
ABD-semantics for $\abdset$ and $I$. We also show that there exists
sets of \abds, even with density 1, that can't be expressed under 
inference-based semantics.
Next, we found  that $\ucq$ query-evaluation problem
is tractable when $\ad{\abdset}=1$
and it is $\conp$-complete for $\ad{\abdset}>1$.
In case $\ad{\abdset}=1$
we prove that there exists an exact table
representation for the ABD-semantics (following that 
this representation can be used for the inference-based semantics too).
This representation is computable with a new chase-based process.
We call this table 
 {\em universal representative} and we believe it is a good candidate to 
be materialized on the target schema 
as the result of the data-exchange process.
We show that for $\ad{\Sigma}=1$
the evaluation of universal queries is tractable 
and introduce a large subclass of $\cqn$
for which the evaluation problem is tractable.
To the best of our knowledge, none of the data-exchange semantics
has tractable results for any subclasses of $\cqn$ containing 
at least one negated atom.

\medskip
{\bf Organization}.
We start with preliminary notions, 
followed by  
Section~\ref{desemantics} in which we 
overview existing semantics and introduce
the new inference-based semantics.
Next section is allocated to present the new mapping 
language based on \abds. 
Section \ref{universal} introduces
 a new
type of na\"ive table called  semi-na\"ive, 
able to exactly represent the ABD-semantics specified by mappings with
density~1.
Section \ref{query} is devoted to 
the problem of certain answers evaluation. 
%For the readers convenience the attached appendix contains the 
%sketch proofs of the main theorems from this paper.

\bigskip
\section{Preliminaries}\label{prelim}
This section reviews the basic technical preliminaries 
and definitions. More information on relational database theory
can be obtained from 
\cite{DBLP:books/aw/AbiteboulHV95}.
We will consider the complexity classes  
$\ptime$, $\np$ and $\conp$. 
For the definition of these classes we refer to
\cite{DBLP:books/daglib/0072413}.

A finite mapping $f$, where $f(a_i)=b_i$, 
for $i\in \{1,\ldots,n \}$, will
be represented as $\{ a_1/b_1, a_2/b_2,\ldots, a_n/b_n \}$.
When it is clear from the context $f$,
it will be also viewed as the following formula
$a_1=b_1 \wedge a_2=b_2 \wedge \ldots \wedge a_n=b_n$.
For a mapping $f$ and set $A$, with $f|_{A}$ will denote
the mapping $f$ restricted to values from $A$.
By abusing the notation, a vector $\vect{x}=(x_1,x_2,\ldots,x_n)$
will be often viewed as the set $\{x_1,x_2,\ldots,x_n \}$,
thus we may have set operations like $x_i \in \vect{x}$
or $\vect{x} \cap \vect{y}$. 
%For a set $A$, $|A|$ denotes its cardinality.

\noindent
{\bf Databases}. A {\em schema} $\schema{S}$ is a finite set 
$\{ S_1, S_2, \ldots, S_n \}$ of relational symbols, each 
symbol $S_i$ having a fixed arity $\ar{S_i}$.
Let $\cons$, $\nulls$ and $\vars$ 
be three countably infinite sets of
constants, nulls and variables
such that there are no common elements between any two 
of these sets.
Elements from $\cons$ are symbolized by lower case 
(possibly subscripted)
characters from the beginning of the alphabet (e.g. $a$, $b_1$).
Elements from $\vars$ are represented by lower case 
(possibly subscripted)
characters from the end of the alphabet (e.g. $z$, $x_2$).
Each element from the countable set $\nulls$ 
is represented by subscripted symbol~$\nullv$ (e.g. $\nullv_i$).
A {\em na\"ive table} $T$ of $\schema{S}$ is an  interpretation
that assigns to each relational symbol $S_i$ a finite
set $S_i^{T} \subset (\cons \cup \nulls)^{\ar{S_i}}$, sometimes
we also view $S_i$ as a relation between elements of $\cons \cup \nulls$.
The set $\dom{T}$ means all elements 
that occur in $T$, clearly $\dom{T}\subseteq \cons\cup \nulls$.
A na\"ive table $T$ is called an {\em instance}
if $\dom{T}\subset \cons$. 
In contrast to general na\"ive tables,
which are identified by capitalized characters from the end of the alphabet
(e.g. $T$, $V$), instances are represented by capitalized 
characters from the middle of the alphabet (e.g. $I$, $J$).
The set of all instances over schema $\schema{S}$ is denoted
$\insta{S}$.
A {\em valuation } $v$ is a mapping over the set $\cons\cup \nulls$ such that 
$v(a)=a$, for all $a\in \cons$, and $v(\nullv)\in \cons$, for all 
$\nullv \in \nulls$. Valuations are extended to
tuples and na\"ive tables as follows. 
For each tuple $\vect{t}=(t_1,t_2,\ldots,t_n)$,
let $v(\vect{t})\defined (v(t_1),v(t_2),\ldots,v(t_n))$;
and for a na\"ive table $T$ over schema $\schema{S}$, define
$v(T)$ as 
 $R^{v(T)}\defined \{ v(\vect{t}) \setdef \vect{t} \in R^{T} \}$,
for all $R\in \schema{S}$. The interpretation
of a na\"ive table $T$ is given by 
$\repe{T} \defined \{ v(T)\setdef \;v \mbox{ valuation}  \}$.

\noindent 
{\bf Schema mappings}.
A data-exchange {\em schema mapping} is a triple 
$\smap{M}=(\schema{S},\schema{T},\Sigma)$,
where $\schema{S}$ and $\schema{T}$ are two disjoint schemas
named the source and target schema respectively;
$\Sigma$ is a set of formulae expressing the relationship
between the source and the target database.
Most commonly, $\Sigma$ is represented by a set of 
source-to-target tuple-generating dependencies (\sttgds) and
target equality-generating dependencies (\egds).
Where a source-to-target tuple-generating dependency is a 
$\fo$ sentence $\xi$ of the form:
$\forall \vect{x}\; (\forall \vect{y}\; \alpha(\vect{x},\vect{y}) 
\rightarrow \exists \vect{z}\; \beta(\vect{x},\vect{z}))$,
where $\vect{x}$, $\vect{y}$ and $\vect{z}$ are vectors of 
variables from $\vars$;
$\alpha(\vect{x},\vect{y})$ (often referred to as the body of the \tgd)
is a conjunction of atoms
over the source schema; 
and $\beta(\vect{x},\vect{z})$ (often referred to as the head of the \tgd)
is a conjunction of atoms over the target schema.
In case $\beta$ contains a single atom, the \tgd\; is referred to 
as GAV (global as view) \tgd.
In case $\vect{z}=\epsilon$ (empty), the \tgd\; is called a 
{\em full } \tgd. 
%In case 
%$\vect{y}=\epsilon$, the \tgd\; is called a 
%{\em quasi-guarded } \tgd.
We will often view a conjunction of atoms
as a na\"ive table where each atom from the conjunction is 
a tuple 
and each variable from the conjunction corresponds to a null in the 
table.
An {\em equality-generating} dependency is a $\fo$ sentence $\xi$
of the form:
$\forall \vect{x}\; (\alpha(\vect{x}) 
\rightarrow x=y)$,
where $\alpha(\vect{x}) $ is a conjunction of atoms over the target
schema and $x$,$y$ are variables from the vector~$\vect{x}$.
A source instance $I\in \insta{S}$ and a target instance
$J \in \insta{T}$ are said to satisfy a \sttgd\; $\xi$,
denoted $(I,J)\models \xi$;
if $I\cup J$ is a model of $\xi$ in the model-theoretic sense.
Similarly, a target instance $J$ satisfies an \egd\; $\xi$, denoted 
$J \models \xi$, if $J$ is a model for $\xi$ in the 
model-theoretic sense.
This is extended to a set of \sttgds\; and \egds\; $\Sigma$
by stipulating that $(I,J)\models \xi$, for all \sttgd\; $\xi\in \Sigma$
and $J\models \xi$, for all \egd\; $\xi\in \Sigma$.
When the schemas are known or not relevant in the context,
we usually interchange the notion of schema mapping and 
the set of dependencies that defines it.

A data-exchange semantics $\mathcal{O}$ associates 
for a schema mapping $\smap{M}=(\schema{S},\schema{T},\Sigma)$
and a source instance $I$ a possible infinite
 set of target instances
$\sol{I}{M}{\mathcal{O}}$. We refer to each element of 
$\sol{I}{M}{\mathcal{O}}$ as a solution for $I$ and $\smap{M}$
under semantics $\mathcal{O}$.

\medskip 
\noindent 
{\bf Queries}.
$\cq$, $\ucq$, $\ucqn$ and $\ucqin$
denote the classes of {\em conjunctive queries}, 
{\em union of conjunctive queries}, 
{\em union of conjunctive queries with negation}
and {\em union of conjunctive queries with unequalities},
respectively. For complete definitions of these classes, please 
refer to~\cite{DBLP:books/aw/AbiteboulHV95}.
For a given data-exchange semantics $\mathcal{O}$,
schema mapping $\smap{M}$ and source instance $I$,
the certain answers for a given query $\query{q}$ is defined as:
\begin{equation}
\begin{split}
\cert{q}{I}{M}{\mathcal{O}} \; \defined \bigcap_{J \in \sol{I}{M}{\mathcal{O}}}  \query{q}(J).
\end{split}
\end{equation}

\section{On data-exchange semantics}\label{desemantics}
As mentioned in the introduction, there are many
semantics considered in data-exchange.
In this section we briefly review the most prominent of these 
semantics and present some certain-answers 
anomalies associated with these.
In Section \ref{inf-sem-sec} 
we introduce a new semantics
for mappings specified by sets of \sttgds\; and \egds\;
that addresses all certain-answers anomalies presented in this paper.
In the final part of this section we
will review the semantics based on 
bidirectional constraints.

\bigskip
\subsection{OWA Data Exchange}
The OWA-Semantics is the first semantics considered in data exchange 
\cite{DBLP:journals/tcs/FaginKMP05}. This is, by far, the
most studied  
\cite{DBLP:journals/tcs/FaginKMP05,DBLP:conf/pods/ArenasBLF04,DBLP:journals/tods/FaginKP05,DBLP:conf/pods/DeutschNR08,DBLP:conf/pods/ArenasPR11,DBLP:conf/icdt/GrahneO12,DBLP:conf/icdt/Hernich12}. 
Under this semantics, given a data-exchange mapping
specified by
$\Sigma$ a set of \tgds\; and \egds\;
and given a source instance $I$,
the OWA-semantics for $I$ under $\Sigma$ is defined as
\begin{equation}\label{sol-owa}
\begin{split}
\hspace{0cm}\sol{I}{\Sigma}{owa} \;\defined\; \{ J\in \insta{T} \setdef 
I\cup J \models 
\Sigma \}.
\end{split}
\end{equation}
\noindent

In their seminal paper Fagin et al. \cite{DBLP:journals/tcs/FaginKMP05}
show that, when dealing with conjunctive queries,
the set of certain answers $\cert{q}{I}{\Sigma}{owa}$
can be computed by na\"{i}vely evaluating \cite{DBLP:dblp_conf/pods/GheerbrantLS13}
 the query $\query{q}$
on a special instance, called {\em universal solution}, 
obtainable in polynomial time, if exists, 
through the chase process.

\begin{example}\label{owa-sample}
Consider source schema consisting 
of two relations: $FTEmployees$ and $Consultants$.
The first relation maintains a list of full-time employees 
and the second relation maintains a list of all consultants from a company. 
Consider target schema consisting 
of relations $AllEmp$ and $Cons$  for all employees and consultants.
The mapping between the source and target represented by 
the following \sttgds:
\begin{equation*}\label{sol-owa_eg}
\begin{split}
 \Sigma= \{ & FTEmployees(eid) \rightarrow AllEmp(eid); \\
& \hspace{-0cm}Consultants(eid)  \rightarrow Cons(eid),AllEmp(eid) \}. 
\end{split}
\end{equation*}

\noindent
Let source instance $I$ be $FTEmployee^{I}=\{ dan \}$ and
$Consultants^{I}=\{ john  \}$. Under this settings, 
instance
 $J \in \sol{I}{\Sigma}{owa}$, where 
$AllEmp^{J}=\{ john, dan \}$ and
$Cons^{J}=\{ john, dan \}$.
Thus, the certain answer to the query: ``{\em Are there any employees that are not consultants?}'', 
expressible by 
a simple $\cqn$ query, will return \false. 
This is an unexpected result
as clearly, based on the source instance,
``$dan$`` is an full-time employee and not a consultant.
\end{example}

In the previous example we note that the
given $\cqn$ query 
will return \false\; for any source instance $I$. 
This uniformity
in the certain query answering was first noted by Arenas et al.
in \cite{DBLP:conf/pods/ArenasBLF04},
who have also shown that
even for the ``copy'' dependencies there are source $\fo$ queries 
that can't be rewritten as target $\fo$ queries that return the same
results under OWA-semantics for all input sources. 

To avoid some of the anomalies, like the one from the previous example,
in~\cite{DBLP:conf/pods/FaginKP03} Fagin et al. proposed a more restricted 
semantics:
\begin{equation}\label{sol-rowa}
\begin{split}
\hspace{0cm}\sol{I}{\Sigma}{rowa} \defined \{ J\in \repe{T} \setdef 
 T& \mbox{ universal solution  for $I$ and $\Sigma$} \}.
\end{split}
\end{equation}

As shown in \cite{DBLP:journals/corr/abs-1107-1456}, 
even with this restriction the certain answers still preserve 
some of OWA anomalies with respect to certain answers.

\medskip
\subsection{CWA Data Exchange}

To outcome the counter-intuitive behavior
under the open-world semantics
Libkin \cite{DBLP:conf/pods/Libkin06} introduced
the CWA-semantics for mappings specified by 
a set of \sttgds. Hernich and Schweikardt
\cite{DBLP:conf/pods/HernichS07} extended the semantics
by adding target \tgds.
Given a set $\Sigma$  of \sttgds\;
and a source instance 
$I$ a CWA-solution for $I$ and $\Sigma$ 
is defined
as any na\"ive table $T$ over the target schema 
that satisfies the following three requirements:
\begin{enumerate}
\item each null from $\dom{T}$ is justified by some tuples from
$I$ and a \sttgd\; from $\Sigma$; and
\item each justification for nulls is used only once; and
\item each fact in $T$ is justified by $I$ and $\Sigma$.
\end{enumerate}

\begin{example}\label{CWA-sample}
Consider mapping specified by the \sttgd\; 
stating that
each student that has a library card reads at least one book:
\begin{equation}\label{tgd-cwa}
\begin{split}
\hspace{0cm} LibraryCard(sid,cid) \rightarrow 
\exists bid\;
Read(sid, bid). 
\end{split}
\end{equation}
\noindent
Let source instance $I$ be  
$LibraryCard^{I}=\{(john,stateLib1),(john,univLib2) \}$,
reviling that
student ``john'' has a library card at the 
municipal library ``stateLib1''
and one card for the university library ``univLib1''.
A CWA-solution for this  settings is
$Read^{T}=\{(john,\nullv_1),(john,\nullv_2)  \}$.
\end{example}

 Libkin \cite{DBLP:conf/pods/Libkin06} considered
4 query-answering semantics: {\em certain}, {\em potential certain},
{\em persistent maybe} and {\em maybe}. 
In this paper we will focus on certain-answer semantics.
More detailed information on the other semantics can be obtained from
 \cite{DBLP:conf/pods/Libkin06,DBLP:journals/tods/HernichLS11}. 
The CWA-semantics for certain-answers is defined as:
\begin{equation}\label{cwa-sol}
\begin{split}
\hspace{-0.1cm}\sol{I}{\Sigma}{cwa} \defined \{ J\in \repe{T} \setdef 
T & \mbox{ is a CWA-solution for $I$ under $\Sigma$} \}.
\end{split}
\end{equation}
Returning to the settings from Example \ref{CWA-sample}, 
clearly, under this semantics  the number of books 
read by a student depends on the number of library cards the student owns.
Thus in our case any target instance
from $\repe{T}$, where $T$ is a CWA-solution, will 
mention that student ``john'' read maximum two books.
Following that the certain answer to the query:
"{\em Did 'john' read exactly one or two book?}"
will be \true.
Clearly this is counter-intuitive as the dependencies do not
exclude that ``john'' read more than 1 book from each library.

Later on, this semantics was improved by designating 
attributes from the mapping as either open or 
closed~\cite{DBLP:journals/jcss/LibkinS11}. 
For example, consider the following 
annotated
\tgd\; that maps each graduated 
student with a course and with a grade for that course:
\begin{equation*}\label{cwa-owa}
\begin{split}
\hspace{0cm}Graduate(s) \rightarrow \exists c\exists g\; 
Attend(s^{cl},c^{op}), Grade(s^{cl},c^{op},g^{cl}). 
\end{split}
\end{equation*}
\noindent
Without the open/closed annotation using the  CWA-semantics each
student is considered to have attended exactly one course for which
the student received exactly one grade. 
To be able to 
represent that a student may have attended more than one course,
the course attribute for both $Attend$ and $Grade$ relation 
are marked as open, all the rest being marked as closed. 
In this case, even if the annotation takes care of the previous issue, 
it introduces another problem as the certain-answer semantics 
for source $\{ Graduate(john) \}$  includes instance $J$, where
 $Grade^J=\{ (john,c01,A),(john,c02,A) \}$ and 
$Attend^J=\{(john,c01) \}$.
Consequently, the certain answer to the query
 "{\em Does ``john'' have grades only
for the courses he attended?}" will be \false\; 
even if from
the mapping we expect the certain answer to be $\true$ 
for any student.

Without further details, it needs to be mentioned here that
Grahne and O. \cite{Grahne:2011:CWC:1966357.1966360} introduced 
another semantics for data exchange called the 
constructible-solution. This semantics,
 when restricted only on source-to-target \tgds, 
coincides with the CWA-semantics.

\subsection{GCWA$^*$ Data Exchange}

The GCWA$^*$ semantics was inspired from 
Minker's~\cite{DBLP:conf/cade/Minker82} 
GCWA- semantics 
and nicely adapted by 
Hernich~\cite{DBLP:journals/corr/abs-1107-1456} for data exchange.
For a source instance $I$ and 
 $\Sigma$ a set of \sttgds\; and \egds\;
let the set $Sol_{min}(I,\Sigma)$
denote all the subset-minimal target instances $J$
with $I\cup J \models \Sigma$.
With this, the CGWA$^{*}$-semantics is defined
as:
\begin{equation*}\label{sol-gcwa}
\begin{split}
\hspace{-0.3cm}\sol{I}{\Sigma}{gcwa^*} \defined \big\{ J \setdef 
J=\bigcup_{i\letbe 1}^{n}J_n \mbox{ for some $n$, } 
 J_i\in Sol_{min}(I,\Sigma) 
\mbox{ and } I\cup J \models \Sigma \big\}
\end{split}
\end{equation*}
\noindent 
Even if the GCWA$^{*}$ semantics solves
all aforementioned anomalies, it introduces new ones 
exemplified hereafter.

\bigskip
\begin{example}\label{cgwa-eg}

Consider source instance with binary relation $DeptC$ 
for departments and names of 
 consultant employees working in that department
and ternary relation $DeptFTE$ for departments and full-time
employee (name and id) from the given department.
Suppose the company hires all the consultants as full-time
employee, thus the target schema will be the ternary relation
$DeptEmp$ with the same structure as $DeptFTE$. 
The exchange mapping is represented as:   
\begin{equation*}\label{eg-gcwan}
\begin{split}
\hspace{0cm} DeptC(did,name) \rightarrow \exists eid\; DeptEmp(did,name,eid); \\
\hspace{0cm} DeptFTE(did,name,eid) \rightarrow DeptEmp(did,name,eid). 
\end{split}
\end{equation*}
Consider source instance with consultants ``$john$'' and ``$adam$'' 
part of the ``$hr$''
department and full-time employee 
``$adam$'' with employee id $1$ part of
``$hr$''. Let target query be: 
{\em Is there exactly one employee named {\em adam} in {\em hr} department?}. 
Under  the CGWA$^{*}$-semantics
the query will return the counter-intuitive answer \true, even 
if based on the source instance and given mapping one would expect the
answer to be \false, because beside full-time employee ``$adam$'' there  
maybe a consultant named ``$adam$'' in the same ``$hr$'' department too.

%
%
%
%Consider a settings where the
%source instance retains for each of a company employee their 
%children in relations based on their education status
%middle or high-school ($Mschool$ and $Hschool$).
%Each of these relations contains the employee id (ei) 
%and the child name (cn),
%except the high-school relation that also contains the 
%child social security number (ssn). 
%The company wants to export this data into another schema 
%represented by a single binary relation ($Parent$) that contains 
%the eid and the child ssn. This exchange is represented
%by the following schema mapping:
%\begin{equation}\label{eg-gcwa}
%\begin{split}
%\hspace{0cm} Mschool(ei,cn) & \rightarrow \exists ssn\; Parent(ei,ssn); \\
%\hspace{0cm} Hschool(ei,cn,ssn) & \rightarrow Parent(ei,ssn). 
%\end{split}
%\end{equation}
%
%Consider a source instance in which each parent has exactly one
%child in high school and none or many children in 
%middle school.
%Under the CGWA$^{*}$-semantics
%with this source instance the certain answer for the query
%"{\em Does each employee have a single child?}" will 
%return the counter-intuitive answer \true.
\end{example}

\medskip
\subsection{Inference-based semantics}\label{inf-sem-sec}
In order to avoid the certain-answer anomalies
presented in this section and in the introduction,
we present 
a new closed-world semantics for mappings specified 
by a set of \sttgds\; and \egds.
For this let us first introduce a few definitions.

\begin{definition}\label{inference1}
Let 
$\xi:  \alpha(\vect{x},\vect{y})\rightarrow \exists \vect{z}\; \beta(\vect{x},\vect{z})$ be a \sttgd\; and $I$ a source instance.
A set of facts $J'$ is said
to be {\em inferred} from $I$ and $\xi$ with function $f$ 
denoted with $I \xrightarrow{\xi}_f J'$ 
if $f(\alpha(\vect{x},\vect{y}))\subseteq I$ and
$f(\beta(\vect{x},\vect{z}))= J'$.
Tuple $t\in J'$ is said to be {\em strongly inferred} from 
$I$ and $\xi$ with function $f$ if 
$t \in f|_{\bar{x}}(\beta(\vect{x},\vect{z}))$, 
otherwise we say that $t$ is {\em weakly inferred}.
\end{definition}

Intuitively a tuple $t$ 
is strongly inferred from $I$ and $\xi$
with $f$ if $t$ does not depend on the assignment given by $f$
to existential variables from $\xi$.

\begin{definition}\label{inference2}
Let $\Sigma$ be a set of \sttgds, $I$ a source instance and
$J$ a target instance. 
A function $\kappa$ that assigns for each 
$\xi \in \Sigma$ a set of functions $\{ f_1,f_2,\ldots, f_n \}$
such that $I \xrightarrow{\xi}_{f_i} J_i \subseteq J$
is called an {\em inference strategy} for $I$ and $J$ with $\Sigma$.
Note that the function that 
allocates the empty set for each $\xi\in \Sigma$ is an inference strategy
for any $I$ and $\Sigma$.
Given an inference strategy $\kappa$ and $\xi\in \Sigma$, a tuple 
$t\in J$ is said to be: 
\begin{itemize}
\item {\em strongly inferred}
with strategy $\kappa$ for $\xi$ if there exist 
$f\in \kappa(\xi)$ such that $t$ is strongly inferred from
$I$ and $\sigma$ with $f$;
\item {\em not inferred} with strategy $\kappa$ for $\xi$
if there is no $f\in \kappa(\xi)$ with
$I \xrightarrow{\xi}_f J'\ni t$ and
\item {\em weakly inferred} with strategy $\kappa$ for $\xi$ otherwise.
\end{itemize}
A tuple $t$ is said to be {\em inferred} with strategy $\kappa$ for $I$, 
$\Sigma$ and $J$ if there exists $\xi\in \Sigma$ such that
$t$ is strongly or weakly inferred with $\kappa$ for $\xi$.
\end{definition}

\begin{example}
Consider the following abstract set of \sttgds\; $\Sigma=\{ \xi \}$,
where:
\begin{equation*}\label{infeq1}
\begin{split}
\xi\;=\; & R(x,y) \rightarrow \exists z\; S(x,z),T(z,y),T(x,y).
\end{split}
\end{equation*}
Consider also  source instance $I$ and target instance $J$ with
$R^I=\{ (a,b),(c,d) \}$, 
$S^J=\{ (a,a),(a,c),(c,a) \}$ and 
$T^J=\{ (a,a),(a,b),(c,b),(a,d),(c,d) \}$.
With this we have $I\xrightarrow{\xi}_{f_1} J_1$,
where function $f_1=\{ x/a, y/b, z/a \}$, 
$S^{J_1}=\{ (a,a)  \}$ and 
$T^{J_1}=\{ (a,b) \}$.
In this case tuple $S(a,a)$ 
is weakly inferred from $I$ and $\xi$ because 
it is obtained by assigning constant $a$ to the
existentially quantified variable $z$. On the other hand,
 tuple $T(a,b)$ is
strongly inferred from $I$ and $\xi$, even if it can be obtain
by assigning constant $a$ to the existentially quantified variable $z$ in 
the second atom in the head of $\xi$, because
the tuple can also be obtain by valuating only universal variables 
$x$ and $y$ to $a$ and $b$ respectively in the third atom.
Consider also functions:
$f_2=\{ x/a, y/b, z/c \}$, $f_3=\{ x/c, y/d, z/a \}$
and $f_4=\{ x/a, y/b, z/b \}$.
With this we have that for any set 
$A\in \mathcal{P}(\{ f_1,f_2, f_3 \})$ function 
$\kappa(\xi)=A$ is an inference strategy for $I$ and $J$ with $\Sigma$.
On the other hand, for any function $\kappa'$ such that $f_4 \in \kappa'(\xi)$
we have $\kappa'$ is not a valid inference strategy because
tuple $T(b,b) \notin J$.  Consider $\kappa(\xi)=\{ f_1, f_2 \}$,
in this case tuples $S(c,a)$ and $T(a,d)$ are not inferred with strategy 
$\kappa$ for $\xi$;
tuples $T(a,b)$ and $T(c,d)$ are strongly inferred with strategy 
$\kappa$ for $\xi$ and finally tuples
$S(a,a)$,$S(a,c)$, $T(a,a)$ and $T(c,b)$ are 
weakly inferred with strategy 
$\kappa$ for $\xi$.
\end{example}

With this we are now ready to introduce the 
inferred-based semantics.
\begin{definition}\label{inf-sem}
Given a source instance $I$ and $\Sigma$ a set of \sttgds\; and \egds\;
the {\em inference-based semantics} for $I$ and $\Sigma$,
denoted with $\sol{I}{\Sigma}{inf}$, is the set of all target 
instances $J$ for which there exists an inference strategy $\kappa$ such that:
\begin{enumerate}
\item Every tuple $t\in J$ is inferred with $\kappa$ for $I$, $\Sigma$ and $J$;
\item For every $\xi: \alpha\rightarrow \beta \in \Sigma$ and 
every function $f$ with $f(\alpha)\subseteq I$ there exists
$f' \in \kappa(\xi)$ such that $f'$ is an extension of $f$;
\item For every $\xi: \alpha\rightarrow \beta \in \Sigma$ and 
$J_{\kappa,\xi}=\bigcup_{I \xrightarrow{\xi}_g J_g, g\in\kappa(\xi)} J_g$, 
there is no function $f$ with
$f(\beta) \in J_{\kappa,\xi}$ and $f(\alpha) \not \subseteq I$,
 where  $f(\beta)$ contains
at least one weakly inferred tuple with strategy $\kappa$ for $\xi$.
\item $(I,J)\models \Sigma$ in the model-theoretic sense.
\end{enumerate}
\end{definition}

\begin{example}\label{egInf1}
Let us consider the same settings as in Example \ref{inference2}
and consider the following 
inference strategies $\kappa_1(\xi)=\{ f_1, f_2, f_3\}$ and  
$\kappa_2(\xi)=\{ f_3\}$.
Because tuple $S(c,a)\in J$ is  
inferred neither by $\kappa_1$ or $\kappa_2$ it follows that
inference strategies $\kappa_1$ and $\kappa_2$ does not respect 
Condition 1 from the definition.
It is easy to observe that $\kappa_1$ satisfies Condition 2. On the
other hand, for function $f=\{ x/a, y/b \}$ we have that
$f(R(x,y)) \in I$ and
there is no extension in $\kappa_2$, thus it follows that   
$\kappa_2$ does not respect Condition 2. 
The instances under condition 3 are 
$J_{\kappa_1,\xi}=\{ S(a,a), S(a,c), T(a,b), T(c,b), T(c,d),T(a,d)   \}$
and 
$J_{\kappa_2,\xi}=\{ S(a,c), T(c,d),T(a,d)   \}$.
For $J_{\kappa_1,\xi}$ we have function $f=\{ x/a, y/d, z/a \}$
such that $f(\beta) \subseteq J_{\kappa_1,\xi}$ and
both $S(a,a)$ and $T(a,d)$ are weakly inferred tuples with $\kappa_1$
and $\xi$ but $f(\alpha)=\{ R(a,d) \} \not \subseteq I$.
Thus $\kappa_1$ does not respect Condition 3. It is easy to see that
$\kappa_1$ respects this condition.
Finally, we have that $(I,J)\models \Sigma$, thus Condition 4 is 
satisfied for any inference strategy with $I$, $\Sigma$~and~$J$.
\end{example}

Intuitively the first rule from the definition states that all 
tuples in the target instance needs to be inferred from the source instance
and the mapping, this is taking care of the tuples not inferred present under
OWA-semantics. This also allows the same nulls to be matched to different 
constant as long as there exists an inference for each of these.
This takes care of the query anomalies present under CWA-semantics. 
The second condition makes sure that all possible source triggers are fired. 
With this
all tuples that can be inferred will be present in at least in one instance
in the semantics, this solves the anomaly presented in Example \ref{cgwa-eg}
for GCWA$^*$.
The third condition ensures that the assignment of nulls does not 
contradict with the inference strategy used, thus taking care of
the query anomaly from the introduction.
The last condition is needed in order to guarantee that the 
instances from the semantics are models for the \egds.
Need to mention here that by renouncing to Condition 3 from the previous 
definition we obtain the semantics defined in \cite{DBLP:conf/pods/GrahneMO15}
in the context of exchange recovery.

It can be observed that the inference based semantics follows all
the rules mentioned by Hernich \cite{DBLP:phd/de/Hernich2010} except
the semantics closure under logical equivalence.
That is, under inference-based semantics, 
there exists a source instance $I$ and 
two logical equivalent sets of
\sttgds\; $\Sigma_1$ and $\Sigma_2$ such that
$\sol{I}{\Sigma_1}{inf}\neq \sol{I}{\Sigma_2}{inf}$.
With the next example we argue that the closure under logical-equivalence
for closed-world data-exchange semantics is not always a desirable property.

\begin{example}\label{noClosure}
Consider source schema consisting of a unary relation 
$SalesPerson$ for the list of sales persons in a company.
A sales person $S1$ can book trades on behalf of 
the accounts assigned to him
and on behalf of the accounts assigned to a sales person $S2$ if
sales person $S1$ is set to cover
sales person $S2$. 
Consider target schema represented by a
binary relation
$Cover$ that specifies if a sales person covers another sales person.
In this case the mapping between the two schemata is represented 
by the following \sttgds:
\begin{equation}\label{nonLog}
\begin{split}
\hspace{0cm} \Sigma_1=\{& SalesPerson(x) \rightarrow Cover(x,x) \\
\hspace{0cm}& SalesPerson(x) \rightarrow \exists y\; Cover(x,y) \}.
\end{split}
\end{equation}
The first dependency specifies that each sales person covers his own accounts
and the second that each sales person covers one or more sales persons.
Clearly $\Sigma_1$ is logically equivalent to $\Sigma_2$, where:
\begin{equation}\label{nonLog1}
\begin{split}
\hspace{0cm} \Sigma_2=\{& SalesPerson(x) \rightarrow Cover(x,x) \}.
\end{split}
\end{equation}
Let source instance $I$ contain two sales persons `{\em john}` and 
`{\em adam}`. Consider target boolean query:
{\em Does `adam` covers only his accounts?}
It is easy to observe that under any 
closed world semantics $\Sigma_2$ and $I$ 
the query will return certain answer \true.
On the other hand, the same query under mapping $\Sigma_1$
is expected to return \false, as `{\em adam}` may cover 
other sales persons as well.
\end{example}

As we will see in Section \ref{query}, this new data-exchange 
semantics has a few desirable properties when it comes to
certain-answers queries. That is, 
for any source-to-target mapping $\Sigma$,
for any instance $I$ and for any \ucq\; query
$\query{q}$ we have $\cert{q}{I}{\Sigma}{inf}=\cert{q}{I}{\Sigma}{owa}$.
Similarly for any full source-to-target mapping $\Sigma$,
any source instance $I$ and for any \fo\; query $\query{q}$
we have $\cert{q}{I}{\Sigma}{inf}=\cert{q}{I}{\Sigma}{cwa}$.
In Section \ref{newsemantics} we will also show 
the language of annotated bidirectional dependencies
is a good language to represent this semantics.

\medskip
\subsection{Bidirectional constraints}
Arenas et al. in \cite{DBLP:conf/amw/ArenasD014}
considered another approach to the certain answer anomaly problem
by changing the language used to express the schema mapping.
For this, the authors proposed the language
of bidirectional constraints. 
Where a bidirectional constraints is a \fo\; sentence of the form:
$$\forall \vect{x}\;  
\alpha(\vect{x}) \leftrightarrow \beta(\vect{x});$$
\noindent
where $\alpha$ and $\beta$ are $\fo$ formulae 
over atoms from the source and 
target schema respectively with free variables $\vect{x}$. 
If the language of $\alpha$
is $\mathcal{L}_{S}$ and of $\beta$ is $\mathcal{L}_{T}$, 
then we are talking about a
$\langle \mathcal{L}_{S},\mathcal{L}_{T} \rangle$-dependency.
Thus if $\alpha$ is represented as a union of conjunctive queries and 
$\beta$ as a conjunction of atoms, the bidirectional constraint
is a $\langle \ucq,\cq \rangle$-dependency.
With this, given a source instance $I$ and
$\abdset$ a set of 
$\langle \mathcal{L}_{S},\mathcal{L}_{T} \rangle$-dependencies.
the bidirectional semantics is defined as:
\begin{equation}\label{sol-bid}
\begin{split}
\hspace{0cm}\sol{I}{\abdset}{\leftrightarrow} \defined \{ J\in \insta{T} \setdef 
I\cup J \models 
\abdset \}.
\end{split}
\end{equation}
This approach did indeed solve most of the anomalies 
related to the other semantics.
Unfortunately, this is achieved at a high
cost, as even the most common data-exchange problems
became non-tractable. For example, testing if the semantics is empty 
is an $\np$-hard problem~\cite{DBLP:conf/amw/ArenasD014}
even for a set of $\langle \cq,\cq \rangle$-dependencies. 
Beside this, the semantics lacks of properties desirable for any 
closed-world semantics. Consider 
$\abdset=\{ \forall x\; (R(x)\leftrightarrow \exists y\; S(x,y)\wedge T(x,y)) \}$,
source instance $I=\{ R(a) \}$ and target instance 
$J$ where $S^{J}=\{ (a,b) \}$
and $T^{J}=\{ (a,b), (a_1,b_1),(a_2,b_2),\ldots,(a_n,b_n) \}$
with $a_i\neq a$ or $b_i\neq b$, for all $i\leq n$.
Clearly we have  $J\in \sol{I}{\abdset}{\leftrightarrow}$
even if none of the tuples under relation $T$, except $T(a,b)$,
are inferred from source instance $I$ and $\abdset$. 
In order to fix this problem one may have to use a 
$\langle \cq,\fo \rangle$-dependency .

Another issue with bidirectional semantics is 
that there are simple unidirectional mappings for which
neither of the presented closed-world semantics are not expressible using 
bidirectional constraints without changing the target 
schema, as shown in the example below.

\begin{example}\label{nonbidirect}
Consider source schema with binary relations
$PFTE$, for projects and the full-time employees assigned on the project,
and $PT$, that contains the tasks associated with each project. 
Let target schema consist of a binary relation $PE$, for projects and employee
assigned to that project, and binary relation $TM$, for
employee and the task they manage.
Consider source instance $I$, with
$PFTE^I=\{ (hr,adam) \}$ and $PT^{I}=\{ (hr,comp) \}$,
stating that full-time employee `adam` works on the `hr` project and
the 'hr' project consists of one task `comp'.
Consider the following mapping $\Sigma$
stating that each full-time employee working on a project is also an
employee working on the project ($\xi_1$) and
that for each project task there exists an employee working on that project
that manages that task ($\xi_2$).
\begin{equation*}\label{notbidirect}
\begin{split}
 \hspace{-0.2cm} \xi_1 &: PFTE(pid,eid)  \rightarrow PE(pid,eid) \\
  \hspace{-0.2cm} \xi_2 &: PT(pid, tid)   \rightarrow\exists eid\;  PE(pid,eid),TM(eid,tid).   
\end{split}
\end{equation*}

It is easy to observe that for any set  $\abdset$ of 
bidirectional constraints there is no possibility to differentiate 
between the tuples from relation $PE$ as being inferred from
 $PFTE$ or $PT$ source relation. 
Thus, for
$J_1=\{ PE(hr,adam) \}$,
$J_2= J_1 \cup \{ TM(adam,comp) \}$ and
$J_3= J_1 \cup \{ PE(hr,sal), TM(sal,comp) \}$,
we have that 
either
$J_1 \in \sol{I}{\abdset}{\leftrightarrow}$ or 
$J_2 \notin \sol{I}{\abdset}{\leftrightarrow}$ or
$J_3 \notin \sol{I}{\abdset}{\leftrightarrow}$,
where ``sal'' is a consultant not a full-time employee.
On the other hand, we have that $J_1 \notin \sol{I}{\Sigma}{*}$ and
$J_2,J_3 \in \sol{I}{\Sigma}{*}$, for any semantics 
$* \in \{ {\sf cwa},{\sf gcwa^{*}},{\sf inf}  \}$.
\end{example}

\section{ABD-semantics}\label{newsemantics}
In this section we propose a new language
and its corresponding semantics 
for data-exchange. 
The new language is based on annotated bidirectional 
dependencies
and not only can properly express the 
inference-based semantics for any set of
\sttgds\; and safe target \egds,
but also can specify mappings not expressible 
in any of the previous semantics. 
Beside this,
the language has syntactical characteristics 
(density and cardinality) 
that clearly
delimit the complexity classes for different problems 
(solution-existence, solution-check and query evaluation).

First, let us 
introduce a few notions and notations.
Given an instance $I$,
a {\em tuple-labeling function}
$\ell$ is a  mapping that assigns 
a non-empty set of integers to each tuple from $I$.

\begin{definition}\label{abd-def}
Given a source schema $\schema{S}$ and a target schema 
$\schema{T}$, an {\em annotated bidirectional dependency}
{\em (\abd)} is a $\fo$ sentence of the form:
\begin{equation}\label{def-abd}
\begin{split}
 \forall \vect{x}\;  ( \exists \vect{y}\; 
\alpha(\vect{x},\vect{y}))  \leftrightarrow (\exists \vect{z}\; \beta(\vect{x},\vect{z})  ),  
\end{split}
\end{equation}
where $\vect{x}$, $\vect{y}$ and $\vect{z}$ are disjoint vectors
of variables;
$\alpha(\vect{x},\vect{y})$ is a conjunction of atoms
over $\schema{S}$;
$\beta(\vect{x},\vect{z})$ is a conjunction of atoms over $\schema{T}$ 
and each atom from $\beta$ is annotated with an integer.
An {\em annotated equality-generating dependency} {\em (\aegd)} is
a target \egd\; where each atom is annotated with an integer.
\end{definition}

For an \abd\; $\xi$ of the form (\ref{def-abd})
$\body{\xi}$ represents the conjunction 
$\alpha(\vect{x},\vect{y})$.
With $\xi^{\rightarrow}$ is denoted \tgd:
$  \forall \vect{x}\;  ( \exists \vect{y}\; 
\alpha(\vect{x},\vect{y}))  \rightarrow (\exists \vect{z}\; \beta'(\vect{x},\vect{z})  )$,
where $\beta'$ removes the annotations used in $\beta$.
For a set $\abdset$ of \abds\; and \aegds\; 
with $\Sigma^{\rightarrow}$ we denote the set
$\{ \xi^{\rightarrow}\setdef \xi \in \abdset \}$.

For a set $\abdset$ of \abds\; and \aegds\; 
with $\abdsetabd$ we denote the 
subset of $\abdset$ that contains all the \abds\; 
and with $\abdsetaegd$ the set that contains all the \aegd.
Thus we have $\abdset=\abdsetabd \cup \abdsetaegd$.
For a target 
relation symbol $R$ and
 $\xi \in \abdset$
the set $\annot(\xi,R)$ contains all
annotations used for atoms over $R$ in $\xi$.
This notation is extended to $\abdset$:
 $\annot(\abdset,R)\defined\cup_{\xi \in \abdset} \annot(\xi,R)$.

For simplicity the quantifiers will be omitted 
when representing an \abd\; or \aegd.
The conjunctions will be represented with commas.

\begin{example}\label{abd-eg1}
Consider the abstract example
with $\abdset=\{\xi_1, \xi_2, \xi_3 \}$, where
\begin{equation*}\label{eg-gcwa1}
\begin{split}
\xi_1:&  \hspace{1.5cm}R(x,y)  
\leftrightarrow  \annt{T}{1}(x,z),\annt{T}{1}(y,z),\annt{T}{2}(x,y); \\
\xi_2:& \hspace{0.3cm} S(x,x),R(x,x)  \leftrightarrow  
\annt{V}{1}(x); \\ 
\xi_3:& \hspace{0.34cm} \annt{T}{1}(x,y),\annt{V}{1}(x)  \rightarrow  x=y; 
\end{split}
\end{equation*}

Atoms in the head of the rules are superscripted 
with their annotation. For this set 
we have $\abdsetabd=\{ \xi_1, \xi_2 \}$
and $\abdsetaegd=\{ \xi_3 \}$.
Furthermore we have $\annot(\xi_1,T)=\{1,2 \}$, 
$\annot(\xi_1,V)=\emptyset$, $\annot(\abdset,V)=\{ 1 \}$.
\end{example}

\begin{definition}\label{annot_dens}{\em (Annotation Density)}
Given $\abdset$ and target relation $R$
the {\em annotation density} 
 for $R$ in $\abdset$ is 
the maximum number of occurrences of 
$\annt{R}{i}$ in $\abdsetabd$ for some $i$
and is denoted with  $\adr{\abdset}{R}$.
The {\em annotation density} for $\abdset$
is defined as:
\begin{equation*}\label{def-dens}
\begin{split}
\ad{\abdset}\defined max \{ \adr{\abdset}{R} \setdef\; R \in \schema{T} \}.
\end{split}
\end{equation*} 

\end{definition}

\begin{definition}\label{annot_card}{\em (Annotation Cardinality)}
Given $\abdset$ and a target relation $R$
the {\em annotation cardinality} 
for $R$ in $\abdset$ is 
defined as $\acr{\abdset}{R}\defined| \annot(\abdsetabd,R) |$.
The {\em annotation cardinality} for $\abdset$
is defined as: 
\begin{equation*}\label{def-car}
\begin{split}
\ac{\abdset}\defined max \{ \acr{\abdset}{R} \setdef\; R \in \schema{T} \}.
\end{split}
\end{equation*} 

\end{definition}

For $\abdset$ from Example \ref{abd-eg1},
$\adr{\abdset}{T}=2$,
because $T$ is annotated with integer~$1$ two times in $\abdsetabd$.
For $V$ we have $\adr{\abdset}{V}=1$, 
thus for the entire set $\ad{\abdset}=2$.
Similarly, $\acr{\abdset}{V}=1$ because only $\annt{V}{1}$ occurs in
$\abdsetabd$. On the other hand, $\acr{\abdset}{T}=2$ thus
$\ac{\abdset}=2$.

Similarly to the notations in  \cite{DBLP:conf/kr/CaliGK08},
an {\em annotated position} in $\abdset$ is a pair $(\annt{R}{k},i)$, 
where $1\leq i \leq \ar{R}$.
Given an \abd\;
$\xi:\; \forall \vect{x}\; (\exists \vect{y}\; 
\alpha(\vect{x},\vect{y}) \leftrightarrow \exists \vect{z}\;  \beta(\vect{x},\vect{z}))$,
the {\em affected positions} in $\xi$, 
denoted $\aff{\xi}$,
is the set of annotated positions 
where elements of vector $\vect{z}$ occurs in~$\xi$.  
This is extended to a set of \abds\; 
$\aff{\abdset}=\cup_{\xi\in \abdset} \aff{\xi}$.

\begin{definition}\label{safe-egd}{\em (Safe \aegd)}
An \aegd\; $\xi \in \abdset$
is said to be {\em safe} if every variable $y$ that occurs 
at least two times in the $\body{\xi}$ occurs only in positions
other than $\aff{\abdset}$.
The set $\abdset$ is said to be {\em safe} if all \aegds\; from $\abdset$
are safe.
\end{definition}

\noindent
Returning to Example \ref{abd-eg1}, 
$\aff{\abdset}=\{ (\annt{T}{1},2) \}$.
Clearly $\xi_3$ is a safe \aegd\; because variable $x$ occurs in positions
$(\annt{T}{1},1)$ and $(\annt{V}{1},1)$ neither being part of
 $\aff{\abdset}$.

In our semantics we will restrict the set of \aegds\; from the mappings
to be safe. This restriction is necessary in order to avoid unwanted
anomalies in the certain answers. To be more specific, 
consider the following
example.

\begin{example}\label{safeaegd}
Let $\Sigma=\{ \xi_1, \xi_2 \}$. The mapping copies employees from the 
source table $Emp_0$ into the target table $Emp$ by replacing the current
$ssn$ key with a new employee id ($eid$) key.
Where $\xi_1$ copies the employee data from source to target and
$\xi_2$ enforces the primary-key constraint. Note that $\xi_2$ is not safe.
 \begin{equation*}\label{safe_egds_eg}
\begin{split}
\xi_1:&  \hspace{0.3cm}Emp_0(ssn,name) \rightarrow 
\exists eid\; Emp(eid,name); \\
\xi_2:& \hspace{0.3cm} Emp(eid, name_1),Emp(eid, name_2) 
 \hspace{0cm} \rightarrow  name_1=name_2.  
\end{split}
\end{equation*}
\noindent 
Let $I$ be specified by
$Emp_0^{I}=\{ (s1, john), (s2, adam), (s3, adam) \}$ and
let  $J$ be specified by 
$Emp^{J}=\{ (a, john), (b, adam) \}$.
It can be easily verified that instance
$J\in \sol{I}{\Sigma}{owa} \cap \sol{I}{\Sigma}{cwa} \cap 
\sol{I}{\Sigma}{gcwa^*}$.
Thus, the following query expressible in $\fo$: 
$\query{q}$\defined {\em Are there exactly two employee named ``adam''?}
will return the counter-intuitive answer \false.
To avoid such behavior
we decided to not allow unsafe \aegds\; at the cost of reducing
the expressivity of
the mapping language.
\end{example}

In order to define the data-exchange semantics for this new language, 
let us first see what 
it means for a target instance over schema $\schema{T}$ 
to satisfy a set $\abdset$ of \abds\; and \aegds.
For this, let $\schema{T_{\diamond}}$ be the schema
 $\schema{T_{\diamond}}=\{ R_i\setdef i\in \annot(\xi,R)\mbox{ where }\xi\in \abdset\mbox{, } R \in \schema{T}  \}$. 
Intuitively 
$\schema{T_{\diamond}}$ contains a distinct relation name for
each pair of relation name and annotation used in $\abdset$.
Similarly, dependency set $\abdsetn$ over $\schema{T_{\diamond}}$
is computed from $\abdset$ by replacing
 each annotated atom $\annt{R}{i}(\vect{x})$ in $\abdset$ 
with atom $R_i(\vect{x})$.
Given $J\in \insta{T}$ and $\ell$ a tuple-labeling function,
with $J_{\ell}$ is denoted the instance over schema
$\schema{T}_{\diamond}$
such that $R_i(\vect{a})\in J_{\ell}$ iff $R(\vect{a})\in J$
and $i \in \ell(R(\vect{a}))$. Instance $J$ is said
to be {\em inferred} from $I$ and $\abdset$ with
tuple-labeling function $\ell$, 
denoted with $(I,J)\models_{\ell} \abdset$,
iff
for any tuple $R(\vect{a}) \in J$ and for
any $i\in \ell(R(\vect{a}))$
there exists $I'\subseteq I$
and $J'_{\ell} \subseteq J_{\ell}$
such that  $R_i(\vect{a}) \in J'_{\ell}$,  
 $(I',J'_{\ell}) \models \abdsetn$
in the model-theoretic sense
and for no instance 
$J''_{\ell} \subsetneq  J'_{\ell}$ it is that 
$(I',J''_{\ell}) \models \abdsetn$.

\begin{definition}\label{ABD-semantics}
The ABD-semantics for a source instance $I$ and a set $\abdset$
of \abds\; and \aegds\; is defined as:
$\sol{I}{\abdset}{abd}:= \{ J\setdef\; (I,J)\models_{\ell}\abdset\mbox{
for some } \ell \}.$
\end{definition}

Intuitively $\sol{I}{\abdset}{abd}$ contains all target instances
which tuples can be labeled in such a way that it satisfies all
the rules from $\abdset$ wrt. their annotation and each tuple 
is inferred from the source and the dependencies.

\begin{example}\label{abd-model-eg}
Consider source schema
containing 3 relations: 
one for full-time employees ($Emp$), 
one for consultant employees ($Cons$) and
 binary relation ($Proj$) that maps 
projects to their location.
The target schema consists of 
unary relation ($AllEmp$) for all employees
and ternary relation ($EmpP$)
that pairs employees with projects and location.
The following mapping specifies: that each full-time ($\xi_1$) 
and consultant employee ($\xi_2$)
works on at least one project in the location
specified by $Proj$ (``$\leftarrow$'' from $\xi_3$ and $\xi_4$); 
for each project there exists at least one 
full-time and at least one consultant employee that works on that project 
(``$\rightarrow$'' from $\xi_3$ and $\xi_4$)
and
all full-time employees are involved only in projects 
from the same location ($\xi_5$).
 \begin{equation*}\label{abd_model_eg_eq}
\begin{split}
\xi_1: & \hspace{0.2cm}Emp(eid) \leftrightarrow 
\annt{EmpP}{1}(eid,pid,lid), \annt{AllEmp}{1}(eid) \\
\xi_2: & \hspace{0.2cm}Cons(cid) \leftrightarrow 
\annt{EmpP}{2}(cid,pid,lid),\annt{AllEmp}{2}(cid) \\
\xi_3: & \hspace{0.2cm}Proj(pid,lid) \leftrightarrow EmpP^{1}(eid,pid,lid)\\
\xi_4: & \hspace{0.2cm}Proj(pid,lid) \leftrightarrow EmpP^{2}(cid,pid,lid)\\
\xi_5: & \hspace{0.2cm} EmpP^{1}(eid,pid_1,lid_1),EmpP^{1}(eid,pid_2,lid_2) \\ & \hspace{5.6cm} \rightarrow lid_1=lid_2
\end{split}
\end{equation*}
\noindent
Let instance $I$ be 
$Emp^{I}=\{ (e_1), (e_2) \}$, $Cons^{I}=\{ (c_1) \}$
and $Proj^{I}=\{ (p_1,ny), (p_2,hk) \}$.
Consider target instances $J_1$, $J_2$ and $J_3$ 
together with their 
tuple-labeling functions $\ell_1$, $\ell_2$ and $\ell_3$:

\small{ 
\noindent
\begin{minipage}[b]{0.5\linewidth}
\centering
\begin{tabular}{l|l} 
 \ \ \ \ \ \ \ \ \ \ \ \ \ \ $J_1$ & $\ell_1(t)$  \\
\hline 
$EmpP(e_1,p_1,ny)$ & \{ 1 \} \\ 
$EmpP(e_1,p_2,hk)$ & \{ 1 \} \\
$EmpP(c_1,p_2,hk)$ & \{ 2 \} \\
$AllEmp(e_1)$ & \{ 1 \} \\
$AllEmp(c_1)$ & \{ 2 \} \\
\multicolumn{2}{l}{} \\
\multicolumn{2}{l}{} \\
\end{tabular}
\end{minipage}
\begin{minipage}[b]{0.32\linewidth}
\centering
\begin{tabular}{l|l} 
 \ \ \ \ \ \ \ \ \ \ \ \ \ \ $J_2$ & $\ell_2(t)$  \\
\hline 
$EmpP(e_1,p_1,ny)$ & \{ 1,3 \} \\ 
$EmpP(e_2,p_2,hk)$ & \{ 1 \} \\
$EmpP(c_1,p_1,ny)$ & \{ 2 \} \\
$EmpP(c_1,p_2,hk)$ & \{ 2 \} \\
$AllEmp(e_1)$ & \{ 1 \} \\
$AllEmp(e_2)$ & \{ 1 \} \\
$AllEmp(c_1)$ & \{ 2 \} 
\end{tabular}
\end{minipage}
} 
\medskip
\small{ 
\begin{center}
\begin{minipage}[b]{0.45\linewidth}
\centering
\begin{tabular}{l|l} 
 \ \ \ \ \ \ \ \ \ \ \ \ \ \ $J_3$ & $\ell_3(t)$  \\
\hline 
$EmpP(e_1,p_1,ny)$ & \{ 1 \} \\ 
$EmpP(e_2,p_2,hk)$ & \{ 1 \} \\
$EmpP(c_1,p_1,ny)$ & \{ 2 \} \\
$EmpP(c_1,p_2,hk)$ & \{ 2 \} \\
$AllEmp(e_1)$ & \{ 1 \} \\
$AllEmp(e_2)$ & \{ 1 \} \\
$AllEmp(c_1)$ & \{ 2 \} \\
\multicolumn{2}{l}{} \\
\end{tabular}
\end{minipage}
\end{center}
}

Clearly $J_1 \notin \sol{I}{\abdset}{abd}$ 
because it does not satisfy \abd\; $\xi_1$,
offending tuple $Emp(e_2)$.
On the other hand,
even if 
$(I,J_{2,\ell_2}) \models \abdsetn$
we have that
$(I,J_2) \not \models_{\ell_2} \abdset$ 
because $3 \in \ell_2(EmpP(e_1,p_1,ny))$
and no \abd\; infers annotation 3. 
Finally,  for $J_3$ and $\ell_3$ we have
$(I,J_3) \models_{\ell_3} \abdset$.
\end{example}

The following theorem relates the inference-based 
semantics to ABD-semantics showing that the
language of annotated bidirectional dependencies 
is a good language to express the inference-based semantics.
This, together with the results from Section \ref{universal} and 
\ref{query}  will help to develop a framework to 
compute ``universal'' target tables used to 
evaluate general queries over the target schema
under this semantics.

\begin{theorem}\label{infer-ABD}
Let $\Sigma$ be a set of \sttgds\; and safe \egds, then 
there exists a set $\abdset$ of \abds\; and \aegds,
with $\ad{\abdset}=1$, such that
for any instance $I$ we have
$\sol{I}{\Sigma}{inf}=\sol{I}{\abdset}{abd}$.
\end{theorem}

One of the important result of this theorem is that 
the annotation density needed to represent the inference-based
semantics is 1 and, as it will be shown next, this ensures 
a sufficient condition for a ``universal'' target table
to exists.
Clearly the converse of this theorem does not hold, 
for this let 
$\abdset= \{ R(x,y)\leftrightarrow T^1(x),S^1(y) \}$
and consider source instance $R^I=\{ (a,b),(c,d) \}$.
We have $\sol{I}{\abdset}{abd}=\emptyset$, 
on the other hand for any $\Sigma$ 
$\sol{I}{\Sigma}{inf}\neq \emptyset$.
The following example gives us a hint on 
how the set $\Sigma$ from the 
theorem is transformed into~$\abdset$.

\noindent
\begin{example}\label{rewrite}
Consider schema for projects ($P$)
project tasks ($PT$) and task employee ($TE$). To this we add
target unary relation $PR$ for the list of projects on the target schema. 
Consider the set $\Sigma$ containing a single source-to-target dependency:
$P(p,e)\rightarrow \exists t\; PT(p,t),TE(t,e),PR(p)$.
For the set $\abdset=\{ \xi_1,\xi_2 \}$, where
\begin{equation*}\label{eq-rewrite}
\begin{split}
\xi_1:\; P(p,e)& \leftrightarrow \annt{PT}{1}(p,t),\annt{TE}{1}(t,e); \\
\xi_1:\;\hspace{0cm} P(p,e) & \leftrightarrow \annt{PR}{1}(p);
\end{split}
\end{equation*}
\noindent 
we have $\sol{I}{\Sigma}{inf}=\sol{I}{\abdset}{abd}$ for any $I$.
\end{example}

\bigskip
\subsection{ABD Solutions}\label{section-abdsolution}

We will now investigate the membership 
and the existence problems under ABD-semantics.
The complexity of these problems will give us an intuition
on under what circumstances one will be able to compute a 
finite representation of the semantics on the target database.
We found that the complexity for the 
existence problem is directly related to the annotation density
and the membership problem 
is connected to the annotation cardinality.
The  {\em ABD-solution existence} problem 
asks if there exists at least one instance part of the semantics
for the given input instance and mapping:
\begin{center}
\framebox[1\width]{
\begin{tabular}{ll} 
{\sc Problem} &   {\sc Sol-Existence}$_{\sf ABD}$($\abdset$)\\
{\sc Input}: &   $I \in \insta{S}$. \\
{\sc Question}: &   Is $\sol{I}{\abdset}{abd}\neq \emptyset$? \\
\end{tabular}
}
\end{center}

Similarly, the {\em ABD-solution checking} problem
verifies if a target instance is part of
the semantics under the given mapping and input source instance:
\begin{center}
\framebox[1\width]{
\begin{tabular}{ll} 
{\sc Problem} &   {\sc Sol-Check}$_{\sf ABD}$($\abdset$)\\
{\sc Input}: &   $I\in \insta{S}$ and $J\in \insta{T}$. \\
{\sc Question}: &   Is $J \in \sol{I}{\abdset}{abd}$? 
\end{tabular}
}
\end{center}

The following dichotomy result directly relates the 
solution-existence problem to annotation density.

\begin{theorem}\label{solexistence-Theo}
Let $\abdset$ be a set of \abds\; and safe \aegds, then
 {\sc Sol-Existence}$_{\sf ABD}(\abdset)$ problem:
\begin{itemize}
\item can be solved in polynomial time if $\ad{\abdset}=1$,
\item is in $\np$ if $\ad{\abdset}>1$ and there exists $\abdset$
with $\ad{\abdset}=2$ such that the problem is $\np$-hard.
\end{itemize}
\end{theorem}

\noindent
Note that the complexity 
of solution existence problem is not
influenced by the annotation cardinality.  That is
the tractable result for annotation density equal to~1
is maintained for any annotation cardinality. Also, 
as shown in the reduction, the problem is $\np$-hard for 
annotation density 2 even with annotation cardinality~1.

Similarly to the existence
problem, we have a dichotomy result
for the solution check problem.
In this case the complexity delimiting factor
is the cardinality.

\begin{theorem}\label{solcheck-Theo}
Let $\abdset$ be a set of \abds\; and safe \aegds, then
the {\sc Sol-Check}$_{\sf ABD}(\abdset)$ problem:
\begin{itemize}
\item can be solved in polynomial time if $\ac{\abdset}=1$, 
\item is in $\np$ if  $\ac{\abdset}>1$
and there exists a $\abdset$
with $\ac{\abdset}=2$  
such that the problem is $\np$-hard.
\end{itemize}
\end{theorem}

The proof attached in the Appendix \ref{proofs} 
clarifies 
that the complexity of the solution-existence problem depends only
on the annotation cardinality. The $\np$-hardness result is obtained
even with annotation density~1.

\section{Universal Representatives}\label{universal}
As mentioned in the introduction,
 data exchange transforms a database
existing under a source schema into another database
under the target schema.
This means that for a given semantics 
it would be preferable
 to be able to
materialize one or more table representations 
on a target. The materialized table(s)
could later be used to obtain
answers for different queries over the target database.
In \cite{DBLP:journals/tcs/FaginKMP05} 
it was shown that for the OWA-semantics expressed 
as a set of \sttgds\; and target \egds,
there exists a {\em universal solution}
that can be 
represented as a na\"{i}ve table
under the target database. 
This universal solution 
can be used to obtain the certain answers to any
$\ucq$ query. 
%Note that the universal solution
%alone is not enough to get the certain 
%answers for $\ucq$ queries even with one inequality per
%disjunct \cite{DBLP:journals/tcs/FaginKMP05}. 
%It can be easily shown that one may use 
%only a universal solution for these types of queries 
%if the target \egds\; would be safe (we don't consider target \tgds).
Need to mention that Deutsch et al. 
in \cite{DBLP:conf/pods/DeutschNR08}
materialize a set of na\"{i}ve tables 
and Gr\"ahne and O. in \cite{Grahne:2011:CWC:1966357.1966360}
use conditional tables to be able to 
retrieve certain answers for a larger range of 
queries under their semantics.
In this section we will introduce a new type of table
capable of representing, under some restrictions, 
all solutions
for ABD-semantics.
Thus this new table can be used to obtain the 
certain answers for any $\fo$ query.

\begin{definition}\label{naiveTabledef}
Let $\nulls$ be partitioned in two countable infinite sets 
$\nullso$ and $\nullsc$.
A {\em semi-na\"ive table} is a na\"ive table $T$
for which each null is identified as being either from 
$\nullso$ or $\nullsc$. 
The semi-na\"ive table  $T$
 has the following interpretation:
\begin{equation*}
\begin{split}
 \repe{T}& \defined  \big \{J=v(\bigcup_{i\letbe 1}^{n} v_i(T)) \setdef  \mbox{ $n$ an integer, }
 \\
& v \mbox{ valuation over }\nullsc, v_i \mbox{ valuation over }\nullso   \big\}
\end{split}
\end{equation*}
\end{definition}

The nulls from $\nullso$ are called {\em open }
and denoted $\nullo$ (possibly subscripted).
The ones from $\nullsc$ are called {\em closed}
and denoted $\nullc$ (possibly subscripted).

\begin{example}\label{seminaiveExample}
Let $T$ be the semi-na\"ive table with
$R^{T}=\{ (a,\nullo_1,\nullc_1,\nullc_2) \}$.
We have $I_1,I_2,I_3 \in \repe{T}$,
where 
$R^{I_1}=\{ (a,a,b,c) \}$,
$R^{I_2}=\{ (a,a,b,c),(a,b,b,c) \}$
and 
$R^{I_3}=\{ (a,a,b,a),(a,b,b,a), (a,c,b,a) \}$.
On the other hand, for
$R^{I_4}=\{ (a,a,b,c),(a,a,b,d) \}$,
we have that $I_4 \notin \repe{T}$,
because closed null $\nullc_2$ was valued to both $c$ and $d$.
\end{example}

To a semi-na\"ive table $T$ 
we may add a global condition $\varphi^*$,
denoted $(T,\varphi^*)$,
as a conjunction of the form 
$\delta_1 \wedge \delta_2 \wedge \ldots \wedge \delta_n$
and each conjunct $\delta_i$, $1 \leq i \leq n$,
is a disjunction of unequalities over the elements from 
$\dom{T}$. 
Given $v$ a valuation over $\nullsc$ and
$v_1,v_2,\ldots,v_n$ valuations over $\nullso$,
for some integer $n$,
we say that 
$(v,\{ v_1, v_2, \ldots, v_n \})$ satisfies 
 $x \neq y$,
denoted  $(v,\{ v_1, v_2, \ldots, v_n \}) \models (x \neq y)$,
iff:
\begin{itemize}
\item $v(x) \neq a$, when $x\in \nullsc$ and $y=a \in \cons$;
\item $v_i(x) \neq a$, for all $i \leq n$, when  $x\in \nullso$ and $y=a \in \cons$;
\item $v(x) \neq v_i(y)$, for all $i \leq n$, when 
$x\in \nullsc$ and $y \in \nullso$; 
\item $v(x) \neq v(y)$, when 
$x,y \in \nullsc$; and
\item $v_i(x) \neq v_j(y)$, for all $i,j \leq n$, when 
$x,y \in \nullso$.
\end{itemize}

The previous notion is naturally extended to a disjunction
of unequalities and to the conjunctive formula $\varphi^*$ where
each conjunct represents a disjunction of unequalities.
This is denoted $(v,\{ v_1, v_2, \ldots, v_n \}) \models \varphi^*$.
With this we can define the interpretation of $(T,\varphi^*)$~as:
\begin{equation*}
\begin{split}
 \repe{T}& \defined  \big \{J=v(\bigcup_{i\letbe 1}^{n} v_i(T)) \setdef  \mbox{ $n$ an integer, }
 \\
& v \mbox{ valuation over }\nullsc, v_i \mbox{ valuation over }\nullso  \\
& \hspace{3cm} \mbox{ and } (v,\{v_1,\ldots,v_n \}) \models \varphi^*  \big\}
\end{split}
\end{equation*}

\begin{example}\label{seminaiveGlobaleg}
Consider $(T,\varphi^*)$, where $T$ and the instances 
are the same as in Example 
\ref{seminaiveExample} and global condition
$\varphi^* \defined (\nullo_1 \neq a \vee \nullc_2 \neq \nullo_1)$.
It can be verified that we have $I_1, I_2 \in \repe{T,\varphi^*}$  
and $I_3 \notin \repe{T,\varphi^*}$ because the tuple
$R(a,a,b,a)$ in $I_3$ was obtained from valuations 
$v=\{ \nullc_1/b, \nullc_2/a \}$ and $v_1=\{ \nullo_1/a \}$
and $(v,\{ v_1 \}) \not \models \varphi^*$.
\end{example}

The following result shows that for a fixed  $\abdset$ with 
$\ad{\abdset}=1$ and an input source instance, we may compute 
in polynomial time 
a representative for its ABD-semantics.

\begin{theorem}\label{abdsemanticsTheo}
Let
$\abdset$ be a set of \abds\; and safe \aegds\; with $\ad{\abdset}=1$.
Then either 
there exists $(T,\varphi^*)$, computable in polynomial time in the 
size of $I$, such that $\sol{I}{\abdset}{abd}=\repe{T,\varphi^*}$,
or $\sol{I}{\abdset}{abd}=\emptyset$.
\end{theorem}

The pair $(T,\varphi^*)$ from the previous theorem, if it exists, 
is called {\em universal representative}
for $I$ and $\abdset$.
In Theorem \ref{abdsemanticsTheo} the universal representative
is computed using a 3-step chase process:
a) ``$\rightarrow$'', similar with the chase from
 \cite{DBLP:journals/tcs/FaginKMP05}
with the difference that the \abds\; will create new  nulls from $\nullso$;
b) ``$=$'' when an \aegd\; equates two open nulls
 or an open null with a closed null 
in the result, they both will be replaced with a new closed null;
c) ``$\leftarrow$'', the chase in the other direction will be used to 
construct the global condition $\varphi^*$ or it fails in case there 
exists a constant in two positions supposed to be distinct.
For the complete annotated-chase algorithm please refer to Appendix \ref{abdchase}.

\begin{example}\label{chaseexample}
Consider $\abdset=\{ \xi_1, \xi_2, \xi_3 \}$, where
\begin{equation*}
\begin{split}
\xi_1 :\; &  \hspace{0.5cm} S(x,y) \leftrightarrow \annt{K}{1}(x,z),\annt{V}{1}(z,y);  \\
\xi_2:\; & \hspace{0.8cm} R(x) \leftrightarrow \annt{U}{1}(x,y); \\
\xi_3: \; &  \hspace{0.5cm} \annt{U}{1}(x,y),\annt{K}{1}(x,z) \rightarrow y=z.
\end{split}
\end{equation*}
Let instance $I$ be $S^I=\{ (a,b),(c,d) \}$ and $R^I=\{ (a) \}$.
In this case a universal representative 
for $I$ and $\abdset$
is the pair $(T,\varphi^*)$,
where $K^{T}=\{ (a,\nullc_1), (c,\nullo_1) \}$, 
$V^{T}=\{ (\nullc_1,b), (\nullo_1,d) \}$, 
$U^{T}=\{ (a,\nullc_1) \}$, $\varphi^* \defined (\nullc_1 \neq \nullo_1)$.
The open null $\nullo_1$ was obtained by simply chasing $\xi_1$ in 
the right direction with source tuple $S(c,d)$; 
the closed null $\nullc_1$ was obtained by equating the two open nulls
obtained by the chase process in the right direction from $\xi_1$ and 
$\xi_2$ with the source tuples $S(a,b)$ and $R(a)$ respectively.
Finally, the condition $\nullc_1\neq \nullo_1$ is obtained 
by chasing to the  left with tuples $\{ K(a,\nullc_1), V(\nullo_1,d) \}$.
Note that when chasing to the left we consider that null values
may equate to any other value.
\end{example}

In order to give a syntactic restriction for an $\abdset$
that will ensure that the global condition is tautological,
we need to introduce first few notations.

For a set $\abdset$,  with $\abdsetright$  is denoted the set 
of \sttgds\; and \egds\;  obtained from $\abdset$ by removing all
annotations and each \abd\; specified as in (\ref{def-abd})
is replaced with the following \sttgd:
$\forall \vect{x}\; (\forall \vect{y}\; 
\alpha(\vect{x},\vect{y}) \rightarrow \exists \vect{z}\; 
\beta(\vect{x},\vect{z}))$.
A set $\Sigma$ of \sttgds\; is said to be {\em GAV-reducible}
if there exists a set $\Sigma'$ of $GAV$ \sttgds\; such that 
$\Sigma$ is logically equivalent to $\Sigma'$.
A set $\Sigma$ is said to be {\em non-redundant } if none of the 
\sttgds\; from $\Sigma$ contains repeated atoms in the head.

The following proposition gives a necessary and sufficient condition
for a set of \sttgds\; to be GAV-reducible.

\begin{proposition}\label{gav-reducibility}
A non-redundant set $\Sigma$ of \sttgds\; is GAV-reducible
iff for each \tgd\; in $\Sigma$
every existentially quantified variable occurs only in one atom in the
head of the \tgd.
\end{proposition}

With this, we can presents a syntactical condition for a
set of dependencies $\abdset$ to ensures that the universal
 representative
 does not have a global condition for input instance  $I$ and
set $\abdset$.

\begin{proposition}\label{restabdsemanticsTheo}
Let $\abdset$ be a set of \abds\; and safe \aegds\; 
with $\ad{\abdset}=1$
such that $\abdsetright$ is GAV-reducible
and let $I$ be a source instance.
Then either 
there exists a semi-na\"ive table $T$, 
computable in polynomial time in the 
size of $I$ such that $\sol{I}{\abdset}{abd}=\repe{T}$
or $\sol{I}{\abdset}{abd}=\emptyset$.
\end{proposition}

\section{Query Answering}\label{query}
If in the previous section we showed how we can compute universal 
representatives for ABD-semantics, 
we will focus here on when and how 
these representatives can be used to compute certain answers 
for different query classes and check the complexities of such 
evaluations. Based on Theorem \ref{infer-ABD} it
follows that all tractable result presented here
are applicable for mappings specified by \sttgds\; and 
safe \egds\; under the inference-based semantics.

Let us first start by defining the certain answer evaluation problem 
for a schema mapping defined under the ABD-semantics.
For this, let  $\abdset$ be a set of \abds\; and safe \aegds. 
Let $\query{q}$ be a query over the target schema $\schema{T}$.
The query evaluation problem for mapping $\abdset$ and query $\query{q}$
is the following decision problem:
%, where
%$\cert{q}{I}{M}{abd} \defined \bigcap_{J \in \sol{I}{M}{abd}} \query{q}(J)$ :

\begin{center}
\framebox[1\width]{
\begin{tabular}{ll} 
{\sc Problem} &   {\sc Eval}$_{\sf ABD}$($\abdset$,$\query{q}$)\\
{\sc Input}: &    $I\in \insta{S}$ and $\vect{t} \in (\dom{I})^{\ar{\vect{t}}}$. \\
{\sc Question}: &   Is $\vect{t} \in \cert{q}{I}{\abdset}{abd}$? 
\end{tabular}
}
\end{center}

The next dichotomy result tells us that one may search for
tractable query evaluation only for mappings with annotation 
density equal to 1.

\begin{theorem}\label{ucqdichotomy}
Let $\abdset$ be a set of \abds\; and safe \aegds. 
If $\query{q} \in \ucq$,
then {\sc Eval}$_{\sf ABD}(\abdset$,$\query{q})$ problem:
\begin{itemize}
\item is polynomial if $\ad{\abdset}=1$ and one 
may use a universal representative $(T,\varphi^*)$, if it exists, 
to answer the problem and
\item is in $\conp$ if $\ad{\abdset}>1$ and there 
exists a mapping $\abdset$ with $\ad{\abdset}=2$ and $\query{q}\in \cq$
such that the problem is $\conp$-hard.
\end{itemize}
\end{theorem}

Note that if the universal representative 
$(T,\varphi^*)$ exists, 
then certain answers 
for \ucq\; queries
can be computed using the na\"ive query evaluation on
$T$ (see \cite{DBLP:journals/tcs/FaginKMP05}).
The following theorem ensures that for 
any mapping $\Sigma$ of \sttgds\; and \egds\;
we can find a corresponding $\abdset$
that agrees with $\Sigma$ on \ucq\; certain answers for
any source instance.

\begin{theorem}\label{UCQequiv}
Let $\Sigma$ be a set of \sttgds\; and safe \egds\;.
Then there exists a set $\abdset$ 
of \abds\; and safe \aegds\;
such that
for any source instance $I$ and $q\in \ucq$
we have 
$\cert{q}{I}{\Sigma}{owa}=\cert{q}{I}{\abdset}{abd}$.
\end{theorem}

From this and Theorem \ref{infer-ABD} we have the
following corollary that states that
the OWA and inference-based semantics
agree on \ucq\; certain answers.

\begin{corollary}\label{owa-inf-ucq}
Let $\Sigma$ be a set of \sttgds\; and
target \egds, then
$\cert{q}{I}{\Sigma}{owa}=\cert{q}{I}{\Sigma}{inf}$
 for any source instance $I$
and $\query{q}\in \ucq$.
\end{corollary}

From Theorem \ref{ucqdichotomy} we know that if the annotation 
density is 1, then we can compute
in tractable time certain answers for $\ucq$. The following
negative result shows that not all queries are tractable
for mappings with annotation density 1.

\begin{theorem}\label{coNPforUCQN}
There exists a set 
 $\abdset$ of \abds\; with $\ad{\abdset}=1$
and there exists a query $\query{q} \in \cqn$ 
such that the problem
{\sc Eval}$_{\sf ABD}(\Sigma$,$\query{q})$ is
$\conp$-complete.
\end{theorem}

From this it follows that for tractable
query evaluation under ABD-semantics we need to restrict either
the set $\abdset$ or the query class used or both.
In the last part of this section we will present such restrictions 
that ensure tractability for certain answers evaluation.

\begin{proposition}\label{fultdquery}
Let $\abdset$ be a set of \abds\; and safe \aegds\;
such that $\abdsetright$ is a collection of full \sttgds\; and safe 
\egds. Then for any $\fo$ query $\query{q}$ the 
{\sc Eval}$_{\sf ABD}(\abdset$,$\query{q})$
is tractable.
\end{proposition}

Note that in the previous proposition there is no need for
annotation density restriction.
The following corollary can be easily be verified.

\begin{corollary}\label{full}
Let $\Sigma$ be a set of full \sttgds, then 
there exists $\abdset$ such that $\abdsetright$ is a collection of full 
\sttgds\; and $\sol{I}{\abdset}{abd}=\sol{I}{\Sigma}{inf}=\sol{I}{\Sigma}{cwa}$.
\end{corollary}  

With $\ucqino$ is denoted the set of $\ucqin$ queries with
at most one unequality per disjunct.

\begin{theorem}\label{unequality}
Let $\abdset$ be a set of \abds\; and safe \aegds\;
with  $\ad{\abdset}=1$. 
Then for any $\ucqino$ query $\query{q}$
 the 
{\sc Eval}$_{\sf ABD}(\smap{\abdset}$,$\query{q})$ problem
is tractable and one may use only the universal representative, 
if it exists, 
to evaluate the query. 
\end{theorem}

Similar result was also shown for OWA \cite{DBLP:journals/tcs/FaginKMP05} 
and CWA-semantics \cite{DBLP:journals/tods/HernichLS11}.
Because the \aegds\; in the mapping
are safe, we can use the universal representative, if it exists, to 
compute the certain answers for any
$\ucqino$. 
This can be extended to OWA-semantics too.
Thus, if the schema mapping is defined 
as a set of \tgds\; and safe \egds, we may use any universal
solution to evaluate certain answers for 
any $\ucqino$ under OWA-semantics eliminating the need to 
re-chase the source instance for each query.
The same as for OWA-semantics 
in case the $\ucqin$ contains 
at most 2 unequalities per disjunct \cite{DBLP:journals/ipl/Madry05}, the 
certain answers evaluation becomes intractable.

\begin{theorem}\label{unequality2}
There exists a  set
$\abdset$ 
 of \abds\; with $\ad{\abdset}=1$
and there is a conjunctive query with
two unequalities such that 
{\sc Eval}$_{\sf ABD}(\smap{\abdset}$,$\query{q})$ problem
is $\conp$-complete. 
\end{theorem}

The following example shows that the certain answers under
inference-based semantics (and implicit under ABD-semantics)
 may differ from 
the certain answers under OWA-semantics 
even for $\ucqino$ queries (based on Corollary \ref{owa-inf-ucq}
the semantics agree on \ucq\; queries).

\begin{example}\label{unequalities}
Let
$\Sigma=\{ R(x,y) \rightarrow \exists z S(x,z),V(z,y) \}$
and
source let instance $I$ be $R^{I}=\{ (a,b), (c,d) \}$.
The universal representative for the corresponding \abd\; set  
$\abdset=\{ R(x,y) \leftrightarrow \annt{S}{1}(x,z),\annt{V}{1}(z,y)\}$
and for the given source instance $I$ 
is $(T,\varphi^*)$,
with $S^{T}=\{ (a,\nullo_1), (c,\nullo_2) \}$,
$V^{T}=\{ (\nullo_1,b), (\nullo_2,d) \}$ and
$\varphi^{*}\defined (\nullo_1 \neq \nullo_2)$.
Consider  
$\query{q} \;\defined\; \exists x\; \exists y\; \exists z_1\; \exists z_2\; S(x,z_1),V(z_2,y),z_1\neq z_2$.
It can be verified that $\cert{q}{I}{\Sigma}{owa}=\false$
and $\cert{q}{I}{\Sigma}{inf}=\cert{q}{I}{\abdset}{abd}=\true$.
\end{example}

In \cite{DBLP:journals/corr/abs-1107-1456} Hernich showed that
if the mapping is given by a restricted set of \sttgds\; (packed \sttgds), 
then the certain answers evaluation problem may be 
answered in polynomial time for universal queries under the 
GCWA$^*$-semantics.
Where a universal query is one of the form 
$\query{q}(\vect{x}) \defined \forall \vect{y}\; \beta(\vect{x},\vect{y})$,
with $\beta$  a quantifier-free $\fo$ formula over the target schema.
In our next result we show that similar polynomial
time can be achieved under the ABD-semantics even without 
any restriction on the \sttgds\; and also by adding 
safe target \aegds.

\begin{theorem}\label{universalqueries}
Let $(T,\varphi^{*})$ be a universal representative
for some source instance $I$ and a set $\abdset$, 
 of \abds\; and safe \aegds\;with $\ad{\abdset}=1$.
Then there exists a polynomial time algorithm, 
 with input $(T,\varphi^*)$ and $\vect{t} \subset \cons$,
such that for any universal query $\query{q}$
decides if $\vect{t} \in \bigcap_{J \in \repe{T,\varphi^*}} \query{q}(J)$. 
\end{theorem}

From the this theorem and Theorem \ref{abdsemanticsTheo} it directly
follows that the {\sc Eval}$_{\sf ABD}$
problem is polynomial for universal queries.

Let us now take a look at $\cqn$ queries.
Theorem~\ref{coNPforUCQN} showed that the certain answers evaluation
is $\conp$-hard for $\cqn$ queries even if the mapping has the annotation
density 1. Next we will present a subclass of $\cqn$ that
has tractable query evaluation properties for a restricted class of 
\abds\; and safe \aegds.
Let $\cqnp$ denote the subclass of  $\cqn$ such that
each query of this class has exactly one positive atom. 
With this we have the following positive result:

\begin{theorem}\label{conjunctnegation}
Let $\abdset$ be a set of \abds\; and safe \aegds\; with $\ad{\abdset}=1$
and such that $\abdsetright$ is GAV-reducible and 
each \aegd\; does not equate two 
variables both occurring in affected positions.
Then for any $\cqnp$ query
the {\sc Eval}$_{\sf ABD}(\abdset,Q)$
problem is polynomial and  can be decided using 
a universal representative.
\end{theorem}

Intuitively, the restrictions on the mapping language from the previous 
theorem ensure that the universal representative does not have any 
global condition (see Proposition \ref{restabdsemanticsTheo} ), 
it contains only open nulls
and because the query contains only one positive atom it ensures
that the Gaifman-blocks, from the universal representative
that match this atom, are bounded in size by a constant value 
(depending only on $\abdset$).

Table 1 summarize the tractable 
results presented in this paper together with the known tractable 
query evaluation 
results from the other semantics. Note that 
for all semantics we considered only \sttgds\; and target
\egds, even if some of these results also hold for restricted classes of
target \tgds.
%Where with "full $\abdsetright$" is denoted the 
%class of $\abdset$ for which the set $\abdsetright$ contains
%only full \tgds. 
%The "restricted $\abdset$" refers to the closed
%sets $\abdset$ such that $\abdsetright$ is GAV-reducible 
%and  each \aegd is safe and does not equates two 
%variables both from affected positions.

\begin{table*}\label{tb1}
\centering
\caption{Complexity of query evaluation}
\begin{tabular}{l|l|l|l|l|l} \hline
     Query/ & $\ucq$ & $\ucqino$  & universal & $\fo$ & $\cqnp$ \\ 
     Semantics &  &  & queries &   \\ 

\hline 
   & \sttgds\;+ & \sttgds\; + & - & -  & -\\
 OWA     & \egds\; \cite{DBLP:journals/tcs/FaginKMP05}  &  
 \egds\;\cite{DBLP:journals/tcs/FaginKMP05} & - &   \\ \hline
      & \sttgds\; + & \sttgds  & - & full \sttgds & - \\
   CWA   &  \egds\; \cite{DBLP:journals/tods/HernichLS11} &  
 \egds\; \cite{DBLP:journals/tods/HernichLS11} & - & 
 \egds \cite{DBLP:journals/tods/HernichLS11}  \\ \hline
      & \sttgds\; + & - & packed & full \sttgds\; + & -\\
  GCWA$^*$    & \egds\; \cite{DBLP:journals/corr/abs-1107-1456} &  & 
\sttgds\; \cite{DBLP:journals/corr/abs-1107-1456} &   \egds\; 
\cite{DBLP:journals/corr/abs-1107-1456} \\ \hline
inference-based/    &  \abdset & \abdset & closed \abdset & full \abdsetright & 
restricted \abdset \\
ABD  &  Theorem \ref{ucqdichotomy} & Theorem \ref{unequality} & Theorem \ref{universalqueries} & Proposition \ref{fultdquery} & 
Theorem \ref{conjunctnegation} \\  \hline
\hline
\end{tabular}
\end{table*}

\section{Conclusions}\label{conclusions}
In this paper we introduced two new 
semantics. One of them (inference-based semantics) 
relying on  \sttgds\; and
target \egds\; language and the second one based on
richer language of annotated bidirectional dependencies.
We showed that the inference-based semantics solves
most of the certain-answers anomalies existing in the 
existing semantics and using the language of \abds\; one may compute
a universal representative that exactly represents the semantics. 

As shown, the language based in \abds\; is much more expressive than 
the one based on \tgds, as we could express the later one using 
only \abds\; with density 1. 
Thus, the work presented here is only the first step for
a full understanding of this new language and semantics.
Even with this, one may be interested in considering 
target \abds, to further increase its expressibility.
For the certain answer semantics it remained an open problem
if one can evaluate certain-answers for any $\cqnp$ queries
in polynomial time for any $\abdset$ and not only 
for the restricted
class of dependencies presented here.
%Another interesting aspect of this semantics, not covered here,
%is the query evaluation under other certain answer semantics.
%Other schema mapping related problems, 
%like composition and inverse, 
%would be interesting 
%to be evaluated under this new mapping language and semantics.
%For example, we believe that the inverse problem will 
%become easier under dependencies with annotation density 1.

\bibliography{Bibliography}

\newpage
\appendix

\section{ABD Chase}\label{abdchase}

In this section we will describe the chase process for a given
source instance $I$ and $\abdset$ a set of \abds\; and safe \aegds.
The {\em annotated chase} algorithm can be represented as a sequence of 
three steps that will result in computing a pair $(T,\varphi^*)$:
\begin{enumerate}
\item ``$\rightarrow$'' chase;
\item ``\egd'' chase; and
\item ``$\leftarrow$'' chase.
\end{enumerate}

1) The ``$\rightarrow$'' chase  constructs a table $T_1$ and a
tuple-labeling function $\ell$
by chasing the \tgds\; from $\abdsetright$ with $I$, 
similarly to the 
the oblivious chase algorithm \cite{DBLP:conf/kr/CaliGK08} with the
difference that the new nulls created are from $\nullso$
and the labeling function maps the 
generated tuple with the annotation of the atom in the dependency 
that generated that tuple.

2) The ``\egd'' 
chase step takes the table $T_1$ and tuple-labeling function $\ell$
obtained from the previous step and either 
change it into table $T_2$ or will fail the chase algorithm.
Let $\xi_e$ be an \aegd\; of the form:
$\alpha(\vect{x}) \rightarrow x=y$,
where $\alpha$ is a conjunction of annotated atoms from the target schema 
and both variables $x$ and $y$ occur in $\vect{x}$. 
Given a target na\"ive table $T$ and a labeling function $\ell$, 
a trigger $\tau$ for $T$ and $\xi_e$
is a pair $(h,\xi_e)$, where $h$ is a homomorphism 
such that $h(\alpha(\vect{x}))\subseteq T$ respecting the annotations 
given by $\ell$.
An \egd\; chase step with trigger $(h,\xi_e)$
is said to fail, denoted $T \xrightarrow{(h,\xi)}_{\ell} \nullv$,
if $h(x)\neq h(y)$ and both $h(x),h(y)\in \cons$.
A non-failing \egd\; chase step
transforms instance $T$ into $T'$, 
denoted  $T \xrightarrow{(h,\xi)}_{\ell} T'$,
where 
\begin{itemize}
\item $T'$ is equal with $T$, if $h(x)=h(y)$; 
\item $T'$ is obtained from $T$ by replacing each occurrence of $h(x)$ 
and $h(y)$ with constant $a$, if either $h(x)=a$ or $h(y)=a$,
this change is reflected as well to $\ell$;
\item $T'$ is obtained from $T$ by replacing each occurrence of 
$h(x)$ and $h(y)$ with new closed null from $\nullsc$,
otherwise.
This change is reflected as well to $\ell$.
\end{itemize}
\noindent
Starting with $T_1$ the previous steps are applied either until
the table is not changed, or if one egd step fails.
In case it fails we say that the annotated chase fails.
Otherwise, let $T_2$ be the instance obtained in this step.

\medskip
3) The ``$\leftarrow$'' 
chase step takes semi-na\"ive table $T_2$ and labeling function 
$\ell$ computed in previous steps and either fails or outputs the 
global condition $\varphi^*$.
For this given an \abd\; $\xi$ of the form:
$\alpha(\vect{x},\vect{y}) \leftrightarrow \beta(\vect{x},\vect{z})$,
let $\xi^{\leftarrow}$ denotes the following sentence:
$\beta(\vect{x},\vect{z}) \rightarrow \exists \vect{y}\; \alpha(\vect{x},\vect{y})$.
For a set $\abdset$, we denote $\abdsetleft$ the set of annotated 
target-to-source \tgds\; obtained by replacing each bidirectional 
dependency with the unidirectional dependency as mentioned before.
A ``$\leftarrow$'' trigger, for $T_2$ and a target-to-source annotated 
\tgd\; $\xi^{\leftarrow}$ is a pair $(H,\xi^{\leftarrow})$,
where $H$ is a set of homomorphisms $\{ h_1, h_2, \ldots, h_k\}$,
$k$ represents the number of atoms in the body of $\xi^{\leftarrow}$,
such that for each atom $R_i(\vect{x})$ from the body of 
$\xi^{\leftarrow}$ we have
$h_i(R_i(\vect{v})) \subseteq T_2$, the annotation of 
atom $R_i(\vect{v})$ in $\xi$ is in 
$\ell(h_i(R_i(\vect{v})))$
and for all variables $x \in \vect{x}\cup\vect{z}$
we have either:
\begin{itemize}
\item $h_i(x)=h_j(x)$ for all $i,j \in \{1, \ldots,k \}$, or
\item if $h_i(x)\neq h_j(x)$, for some distinct $i,j \in \{1, \ldots,k \}$,
then either  $h_i(x)$ or $h_j(x)$ is a null.
\end{itemize}

For the set $H$ of homomorphisms the mapping $h_H$ is defined
such that for each 
$x \in \vect{x}\cup\vect{z}$:
\begin{equation*}\label{anntchase}
\begin{split}
h_H(x)= 
\begin{cases}
h_1(x),  & \mbox{if } h_1(x)=h_2(x)=\ldots=h_k(x) \\
h_j(x),  & \mbox{otherwise, where $h_j(x)$ is a constant} \\
& \mbox{for some $j \in \{1,\ldots, k \}$,} \\
& \mbox{or $\forall i \in \{1,\ldots, k \}$,  $h_{i}(x)\notin \cons$}
\end{cases}
\end{split}
\end{equation*}

We say that $(H,\xi^{\leftarrow})$ generates
na\"ive table $T'$, 
denoted $T_2 \xrightarrow{(H,\xi^{\leftarrow})} T'$,
if $T'=h'(\alpha(\vect{x},\vect{y}))$, for some 
extension $h'$ of $h_H$ that assigns a new null value
for each variable from $\vect{y}$.
If there is no homomorphism from
table $T'$ into $I$,
then let 
$\varphi_{(H,\xi^{\leftarrow})} \defined \neg (\bigwedge_{x\in \vect{x}} 
\bigwedge_{1\leq i,j \leq k} h_i(x)=h_j(x)  ) $.
If $\varphi_{(H,\xi^{\leftarrow})}$ for some trigger $(H,\xi^{\leftarrow})$  
is a contradiction,
 we say that the 
annotated chase algorithm fails.
Otherwise, let $\varphi^*$ be the conjunctions of all formulae 
$\varphi$
constructed in the previous process,
that is:
\begin{equation*}\label{globalCond}
\begin{split}
\hspace{0cm}\varphi^* \defined \bigwedge_{\xi^{\leftarrow} \in \Sigma^{\leftarrow}}\;\;\;
 \bigwedge_{(H,\xi^{\leftarrow})\mbox{ \tiny{trigger}}} \varphi_{(H,\xi^{\leftarrow})}. 
\end{split}
\end{equation*}

Finally, if the algorithm does not fail it will return pair 
$(T_2,\varphi^*)$.

\bigskip 
\begin{example}\label{anntchaseeg}
Consider the following set $\abdset$:
\begin{equation*}\label{exampleantch}
\begin{split}
\xi_1:& \hspace{2cm} R(x,y) \leftrightarrow \annt{S}{1}(x,z),\annt{S}{2}(y,z),\annt{V}{1}(x,z); \\
\xi_2: & \hspace{0.2cm} \annt{V}{1}(x,z_1),\annt{S}{2}(x,z_2) \rightarrow z_1=z_2.
\end{split}
\end{equation*}

Let source instance $I$ be $R^{I}=\{ (a,b),(c,a) \}$.
The  ``$\rightarrow$'' step will
 construct
$T_1$ and tuple-labeling function $\ell$, with 
$S^{T_1}=\{ (a,\nullo_1):1, (c,\nullo_2):1,
(b,\nullo_1):2, (a,\nullo_2):2 \}$, 
and
$V^{T_1}=\{ (a,\nullo_1):1, (c,\nullo_2):1 \}$.
After applying the  ``\egd'' steps we 
obtain semi-na\"ive table 
$T_2$ and corresponding tuple-labeling function $\ell$, 
where
$S^{T_2}=\{ (a,\nullc_1):1, (c,\nullc_1):1, (b,\nullc_1):2, (a,\nullc_1):2 \}$, 
$V^{T_2}=\{ (a,\nullc_1):1, (c,\nullc_1):1 \}$.
Finally, for ``$\leftarrow$'' step we have
that $(H,\xi^{\leftarrow})$
is a trigger for $T_1$, where
$H=\{ h_1=\{\ x/a, z/\nullc_1 \}, h_2=\{ y/a, z/\nullc_1 \},
h_3=\{ x/a, z/\nullc_1 \} \}$. 
With this we obtain 
$\varphi_{(H,\xi^{\leftarrow})}= \neg ( a=a \wedge \nullc_1=\nullc_1 )$ 
which is 
a contradiction. From this it follows that the 
annotation chase will fail.
Note that if we remove $\xi_2$ from $\abdset$, $T_2$ would be 
equal to $T_1$ and the global condition $\varphi^{*}$ would be 
equivalent to 
$\varphi^* \defined (\nullo_1 \neq \nullo_2) $.
\end{example}

\bigskip 
\section{Sketch Proofs}\label{proofs}
In this section we provide sketch proofs for the main
results presented in the paper. The complete proofs will be
provided in the full version of this paper.

In order to show the proof of Theorem \ref{infer-ABD}
we need to introduce first a few notations.

\begin{definition}{\em \cite{DBLP:conf/pods/GottlobN06}}
The {\em Gaifman graph} $G^I$  for a table $T$ is an undirected
graph with vertex set $\dom{T}\cap \nulls$ and an
edge between two vertices $x$ and $y$ if $x$ and $y$ occurs together
in a tuple of $T$. A {\em block} is a connected set of nulls in $G^T$.
\end{definition}

\begin{definition}
Let $T$ be a table.
A set $\{ T_1,T_2,\ldots, T_n \}$ 
is called a {\em Gaifman partition} of $T$
if the following holds:
\begin{itemize}
\item $\{ T_1,T_2,\ldots, T_n \}$ is a partition of $T$; and
\item for each $x \in \nulls(T)$ there exists a exactly one $i$
such that $x \in \nulls(T_i)$; and
\item if nulls $x$ and $y$ are in the same block of $G^T$, 
then there exists exactly one $i$ such that $x,y \in \nulls(T_i)$; and
\item if $t\in T$ and $t$ does not contain any nulls, 
then there exists $i$ such that $T_i=\{ t \}$. 
\end{itemize}
\end{definition}

Clearly for each table $T$ there exists a unique Gaifman partition
of $T$.

\medskip
\setcounter{theorem}{0}
\begin{theorem}
Let $\Sigma$ be a set of \sttgds\; and safe \egds, then 
there exists a set $\abdset$ of \abds\; and \aegds,
with $\ad{\abdset}=1$, such that
for any instance $I$ we have
$\sol{I}{\Sigma}{inf}=\sol{I}{\abdset}{abd}$.
\end{theorem}
\begin{proof}
Construct $\abdset$ from $\Sigma$ as follows:
for each \tgd\; $\xi \in \Sigma$, where 
$\xi: \alpha(\vect{x},\vect{y})\rightarrow \exists \vect{z}\;\beta(\vect{x},\vect{z})$,
add \abds\;
$\alpha(\vect{x},\vect{y}) \leftrightarrow \beta'_i(\vect{x},\vect{z})$,
where $\beta_i(\vect{x},\vect{z})$ is a table from the Gaifman 
partition of $\beta(\vect{x},\vect{z})$ (where each variable from $\vect{x}$
is treated as a constant and variables from $\vect{z}$ are treated as 
nulls) and
$\beta'_i$ is obtained from $\beta_i$ by annotating each relation 
with a distinct integer.
For each \egd\; $\xi \in \Sigma$, where 
$\xi: \alpha(\vect{x},\vect{y})\rightarrow x=y$
add \aegds\;
$\alpha_i(\vect{x},\vect{y})\rightarrow x=y$
to $\abdset$, where $\alpha_i$ is obtained from $\alpha$
by annotating each relation with annotations used in \abds\; for the same 
relation. Note that if an \egd\; from $\Sigma$ contains a relation name not 
occurring in any \tgds, then that \egd\; will not be reflected in $\abdset$.
Because each \tgds\; is mapped to a set of \abds\; in $\abdset$
based on the Gaifman partitioning it follows that 
there always exists a solution for any source instance $I$ and the 
constructed $\abdset$, of course if the equality dependencies are satisfied.
From here it is a simple exercise to verify that all properties
of the inference-based semantics are fulfilled by $\abdset$ under 
ABD-semantics. Thus, it follows that for any source instance $I$,
$\sol{I}{\abdset}{abd}=\sol{I}{\Sigma}{inf}$ 
\end{proof}

\bigskip
For the next theorem we need the following result.

\medskip
\begin{lemma}\label{ThreeTwoColorable}
If a graph is 3-colorable but not 2-colorable, then for any 3 coloring of 
the graph there exists at least an edge between any two distinct colors.
\end{lemma}

\medskip
With this we can now prove the first dichotomy theorem.

\medskip
\setcounter{theorem}{1}
\begin{theorem}
Let $\abdset$ be a set of \abds\; and safe \aegds, then {\sc Sol-Existence}$_{\sf ABD}(\abdset)$ problem:
\begin{itemize}
\item can be solved in polynomial time if $\ad{\abdset}=1$,
\item is in $\np$ if $\ad{\abdset}>1$ and there exists $\abdset$
with $\ad{\abdset}=2$ such that the problem is $\np$-hard.
\end{itemize}
\end{theorem}

\begin{proof}
In case $\ad{\abdset}=1$, from Theorem \ref{abdsemanticsTheo} we know 
that if for a source instance $I$ the universal representative $(T,\varphi^*)$
exists, then $\repe{T,\varphi^*}=\sol{I}{\abdset}{abd}$.
From the construction of $(T,\varphi^*)$ we know that $\varphi^*$ is always 
satisfiable (just assign new distinct constant values for each null).
Thus, the problem  resumes itself to find if such universal 
representative exists. 
This can be done in polynomial time in the size of $I$ using
the annotated chase algorithm from Section \ref{abdchase}.

For the second part it is clear that the solution existence 
problem is in $\np$ for any $\ad{\abdset}>1$
as one may guess a target instance $J$ and labeling function $\ell$
that labels each tuples with a set of integer, the values from the 
set being restricted from a constant 
set given by $\abdset$.
Then test in polynomial time that $(I,J)\models_{\ell} \abdset$.

For the completeness part  we will use a 
reduction from the graph 3-colorability problem known to be 
$\np$-complete.% \cite{Garey:1979:CIG:578533}.
We could not use directly the reduction from 
\cite{DBLP:conf/amw/ArenasD014} because that reduction 
has $\ad{\abdset}=4$. Our reduction not only has
$\ad{\abdset}=2$ but $\ac{\abdset}=1$.
The restriction that $\ac{\abdset}=1$ is important 
in order to show that the annotation cardinality 
does not influence the solution existence problem.
For this let us consider the following set of
\abds\; $\abdset$:
\begin{equation*}\label{eq1:eg:bidirect3}
\begin{split}
&\hspace{0.95cm}  V(x)  \leftrightarrow \annt{B}{1}(x,v),\annt{C}{1}(x,v);\\
&\hspace{0.5cm}  E_0(x,y)  \leftrightarrow \annt{E}{1}(x,y); \\
&\hspace{0.65cm}  D(z,v)  \leftrightarrow \annt{B}{1}(x,z),\annt{C}{1}(y,v),\annt{E}{1}(x,y).
\end{split}
\end{equation*}

\medskip
To this consider source instance $I$ that represents under 
unary relation $V$ all the vertexes of a graph $G$,
under binary relation $E_0$ all the edges in the graph. 
And finally, under binary relation $D$ we'll have all the 
combinations of two distinct colors from the set $\{r,g,b \}$.
Based on the previous lemma, we know that if the graph is 3-colorable
but not 2-colorable, there will be an edge between any 2 distinct colors,
thus the third \abd\; holds if the graph is 3-colorable but not 2-colorable.
This means that our reduction needs to verify first in polynomial time 
if graph $G$ is 2-colorable. If it is, then return \true.
If it is not 2-colorable we create source instance $I$ from $G$ as
mentioned before. 
With this and Lemma~\ref{ThreeTwoColorable}
it is obvious that the graph $G$, not 2-colorable 
is 3-colorable iff the solution existence problem returns \true.
From this it follows that there exists $\abdset$ with $\ad{\abdset}=2$ 
for which the problem is $\np$-complete
\end{proof}

\medskip
\setcounter{theorem}{2}
\begin{theorem}
Let $\abdset$ be a set of \abds\; and safe \aegds, then
the {\sc Sol-Check}$_{\sf ABD}(\smap{M})$ problem:
\begin{itemize}
\item can be solved in polynomial time if $\ac{\abdset}=1$, 
\item is in $\np$ if  $\ac{\abdset}>1$
and there exists a $\abdset$
with $\ac{\abdset}=2$  
such that the problem is $\np$-hard.
\end{itemize}
\end{theorem}
\begin{proof}
In case $\ac{\abdset}=1$, the solution-check problem is the same
as equivalent to check if the instance $J$ together with $I$ 
satisfies a set of $\fo$ formulae, problem known to be tractable 
for a fix set of formulae.
For the second part, given target instance $J$, 
one may guess in polynomial 
time a labeling function $\ell$ for $J$ 
that takes values from
a fixed set of integers determined by $\abdset$. 
It can be verified in polynomial time, in the size of $I$,
if $(I,J) \models_{\ell} \abdset$.
For the completeness we will use a reduction from the 
graph 3-colorability problem as follows. Let $\abdset$
be:
\begin{equation*}\label{eq1:eg:bidirect4}
\begin{split}
&\hspace{0cm} D(z,v) \leftrightarrow \annt{B}{1}(x,z),\annt{C}{1}(y,v),\annt{E}{1}(x,y); \\
&\hspace{0cm} V(x,v) \leftrightarrow \annt{B}{2}(x,v); \\
&\hspace{0cm} V(x,v) \leftrightarrow \annt{C}{2}(x,v). 
\end{split}
\end{equation*}

\medskip
\noindent
Note that $\ac{\abdset}=2$ and $\ad{\abdset}=1$.
Consider source instance $I$ representing binary relation $D$ 
with a pair of each distinct color from the set  $\{r,g,b \}$.
For each vertex $x$ of the graph $G$ and each color $c$ 
from $\{ r,g,b \}$
 binary relation $V$ will contain tuple $(x,c)$.
Finally, target instance $J$ will contain under relations $B$ and $C$
the same tuples as $V^I$ and binary relation $E$ will contain the 
edges from the graph.
It can be verified that 
if there exists a labeling function
$\ell$ such that $(I,J)\models_{\ell} \abdset$,
then
the tuples in $B$ which labeling contains integer $1$ represent 
the 3-coloring mapping for graph $G$. Clearly the converse is also true.
From this it follows that graph $G$ is 3-colorable iff 
$J\in \sol{I}{\abdset}{abd}$
\end{proof}

\medskip
\setcounter{theorem}{3}
\begin{theorem}
Let
$\abdset$ be a set of \abds\; and safe \aegds\; with $\ad{\abdset}=1$.
Then either 
there exists $(T,\varphi^*)$, computable in polynomial time in the 
size of $I$, such that $\sol{I}{\abdset}{abd}=\repe{T,\varphi^*}$,
or $\sol{I}{\abdset}{abd}=\emptyset$.
\end{theorem}
\begin{proof}
The result follows from the construction of $(T,\varphi^*)$
in the annotated chase algorithm presented in Section~\ref{abdchase}.
Thus,  in case the  algorithm
 fails, then there is no
target instance part of the ABD-semantics. 
In case the 
annotated chase algorithm returns 
table $(T,\varphi^*)$, then $\repe{T,\varphi^*}=\sol{I}{\abdset}{abd}$
\end{proof}

\medskip 
\setcounter{proposition}{0}
\begin{proposition}
A non-redundant set $\Sigma$ of \sttgds\; is GAV-reducible
iff for each \tgd\; in $\Sigma$
every existentially quantified variable occurs only in one atom in the
head of the \tgd.
\end{proposition}

\begin{proof}
For the ``if'' direction, consider $\Sigma$ a set of \sttgds\;
such that no existentially quantified variable occurs in two distinct atoms
in the head. We will construct $\Sigma'$ such that for each $\xi \in \Sigma$,
with $\xi$ a sentence of the form: 
$$\alpha(\vect{x},\vect{y})\rightarrow \exists \vect{z}\; R_1(\vect{x},\vect{z}_1),R_2(\vect{x},\vect{z}_2),\ldots,R_k(\vect{x},\vect{z}_k);$$
\noindent
where  $\{ \vect{z}_1, \vect{z}_2, \ldots \vect{z}_k \}$ is a partition of 
$\vect{z}$,
we will add the following $k$ \sttgds\; to $\Sigma'$:
$$\xi'_i \defined \; \alpha(\vect{x},\vect{y}) \rightarrow R_i(\vect{x},\vect{z}_i);$$ 
\noindent 
for all $1 \leq i \leq k$. It is easy to verify that $\Sigma'$ is a set of
GAV dependencies and $\Sigma'$ is logically equivalent to $\Sigma$.

For the ``only if'' direction, consider that there exists $\Sigma$ 
logically equivalent to $\Sigma'$  a set of GAV dependencies 
such that $\Sigma$ contains a \sttgd\; $\xi$ and existentially quantified 
 variable $z$, 
with two distinct atoms in the head of $\xi$
that share $z$. Let these two atoms be $R(\vect{x},z)$ and $S(\vect{y},z)$.
Clearly, these relational symbols occur in the head of some 
\sttgd\; in $\Sigma'$. If $z$ is existentially quantified in both 
dependencies in $\Sigma'$, it follows that there exists a subset minimal 
target instance with those positions different that models a 
source instance $I$ and triggers at least one of those two dependencies.
But that instance $J$ is not a model for $\Sigma$, contradicting with our 
assumption that $\Sigma$ and $\Sigma'$ are logically equivalent.
Similarly, it can be proved if one of the positions for $z$ is 
universally quantified.
In this case an adequate source instance needs to be selected
\end{proof}

\bigskip
\setcounter{proposition}{1}
\begin{proposition}
Let $\abdset$ be a set of \abds\; and safe \aegds\; 
with $\ad{\abdset}=1$
such that $\abdsetright$ is GAV-reducible
and let $I$ be a source instance.
Then either 
there exists a semi-na\"ive table $T$, 
computable in polynomial time in the 
size of $I$ such that $\sol{I}{\abdset}{abd}=\repe{T}$
or $\sol{I}{\abdset}{abd}=\emptyset$.
\end{proposition}
\begin{proof}
It follows directly from Proposition \ref{gav-reducibility} and
the way the annotated chase algorithm
presented in Section \ref{abdchase} constructs the global condition 
$\varphi^*$. 
Because there are no two atoms in any \abd\; that share 
an existential variable, we have either $\varphi^* \equiv \true$
or the annotated chase algorithm fails for the given input
\end{proof}

\bigskip
\setcounter{theorem}{4}
\begin{theorem}
Let $\abdset$ be a set of \abds\; and safe \aegds. 
If $\query{q} \in \ucq$,
then {\sc Eval}$_{\sf ABD}(\smap{M}$,$\query{q})$ problem:
\begin{itemize}
\item is tractable if $\ad{\abdset}=1$ and one 
may use a universal representative $(T,\varphi^*)$, if it exists 
to answer the problem and
\item is in $\conp$ if $\ad{\abdset}>1$ and there 
exists a mapping $\abdset$ with $\ad{\abdset}=1$ and $\query{q}\in \cq$
such that the problem is $\conp$-hard.
\end{itemize}
\end{theorem}
\begin{proof}
For the case $\ad{\abdset}=1$ the proof is similar with the 
proof from \cite{DBLP:journals/tcs/FaginKMP05}
with the observation that for any valuation $(v,\{ v_1 \})$, where
$v$ and $v_1$ map each closed null and open null respectively
to a new distinct constant from $\cons$, then
we have that $v\circ v_1 (\varphi^*) \equiv \true$
and $v\circ v_1(T) \in \sol{I}{\abdset}{abd}$.

In case $\ad{\abdset}>1$, one may guess
 $J\in \sol{I}{\abdset}{abd}$, such that 
$\query{q}(J) = \false$, thus making the problem in $\conp$.
For the completeness part consider the set $\abdset$ defined as:
\begin{equation*}\label{eq1:eg:bidirect5}
\begin{split}
&\hspace{0.41cm} V(x) \leftrightarrow \annt{B}{1}(x,v),\annt{C}{1}(x,v); \\
&\hspace{0.35cm} M(v) \leftrightarrow \annt{C}{1}(x,v); \\
&\hspace{0cm} E_0(x,y) \leftrightarrow \annt{E}{1}(x,y). 
\end{split}
\end{equation*}
Given a graph $G$ we construct in polynomial time source instance 
$I$ such that $V^I$ will contain all the vertexes from $G$,
$M^I$ will contain three tuples, one tuple for each color, and
$E_0^I$ representing the edges in the graph.
To this, consider boolean query
$\query{q}\defined B(x,z),C(y,z),E(x,y)$. 
It is verifiable that  
$\cert{q}{I}{M}{abd}=\true$ iff the graph is not 3-colorable.
\end{proof}

\medskip 
\setcounter{theorem}{5}
\begin{theorem}
Let $\Sigma$ be a set of \sttgds\; and safe \egds\;.
Then there exists a set $\abdset$ 
of \abds\; and safe \aegds\;
such that
for any source instance $I$ and $q\in \ucq$
we have 
$\cert{q}{I}{\Sigma}{owa}=\cert{q}{I}{\abdset}{abd}$.
\end{theorem}
\begin{proof}
Construct $\abdset$ from $\Sigma$ the way it was shown
in the proof of Theorem \ref{infer-ABD}.
From the ABD-chase algorithm 
and from the way the  $\abdset$ was constructed,
it follows that there exists a universal representative $(T,\varphi^*)$
for $I$ and $\abdset$ if and only if there exists 
a universal solution $U$ for $I$ and $\Sigma$.
From the construction of $T$ we have that $T$ and $U$
are homomorphically equivalent and also $\varphi^*$ satisfiable 
(just assign
distinct constant for each null). From this and from 
\cite{DBLP:journals/tcs/FaginKMP05}
 follows that 
$\cert{q}{I}{\Sigma}{owa}=\cert{q}{I}{\abdset}{abd}$
for any $q\in \ucq$
\end{proof}

\medskip

\setcounter{theorem}{6}
\begin{theorem}
There exists a set 
 $\abdset$ of \abds\; with $\ad{\abdset}=1$
and there exists a query $\query{q} \in \cqn$ 
such that the problem
{\sc Eval}$_{\sf ABD}(\Sigma$,$\query{q})$ is
$\conp$-complete.
\end{theorem}
\begin{proof}
Clearly the problem is in $\conp$.
For the hardness result
the reduction is adapted from Hernich's 
\cite{DBLP:journals/corr/abs-1107-1456} reduction from 
the {\sc Clique} problem,
% \cite{Garey:1979:CIG:578533}
that is to decide, 
for an undirected graph without loops, if it contains a clique of
size $k$.
Let $\abdset$ be defined by:
\begin{equation*}\label{eq1:eg:bidirect6}
\begin{split}
&\hspace{0cm} E_0(x,y) \leftrightarrow \annt{E}{1}(x,y); \\
&\hspace{0cm} C_0(x,y) \leftrightarrow \annt{C}{1}(x,y),\annt{A}{1}(x,z),\annt{B}{1}(y,v). 
\end{split}
\end{equation*}
Note that $\ad{\abdset}=1$. Let source instance $I$ being defined such that  $E_0^I$ has all the 
edges from the graph and relation $C^I$ has all pairs of
disjoint elements from $\{ c_1, c_2,\ldots, c_k \} \subseteq \cons$.
Consider boolean $\cqn$ query:
$$\query{q} \defined \exists x\;\exists y\;\exists z_1\; \exists z_2\; 
C(x,y),A(x,z_1),B(y,z_2),\neg E(z_1,z_2).$$
\noindent
It can be verified that  $\cert{q}{I}{\abdset}{abd}=\true$ if
the graph does not contain a clique of size $k$, because for each 
target instance  $J \in \sol{I}{\abdset}{abd}$ 
for any complete graph of size $k$ in $J$
there exists an edge that is not 
in $E$. The converse holds as well, thus proving that 
$\cert{q}{I}{\abdset}{abd}=\true$ iff the graph does not contain a clique
of size $k$
\end{proof}

\setcounter{proposition}{2}
\begin{proposition}
Let $\abdset$ be a set of \abds\; and safe \aegds\;
such that $\abdsetright$ is a collection of full \sttgds\; and safe 
\egds. Then for any $\fo$ query $\query{q}$ the 
{\sc Eval}$_{\sf ABD}(\abdset$,$\query{q})$
is tractable.
\end{proposition}
\begin{proof}
It follows directly from Theorem  \ref{abdsemanticsTheo}, 
the annotation chase algorithm defined in Section \ref{abdchase}
and the observation that during the chase process no new nulls will be created,
thus the result of the abd-chase process it will be an instance, if
exists
\end{proof}

\setcounter{corollary}{1}
\begin{corollary}
Let $\Sigma$ be a set of full \sttgds, then 
there exists $\abdset$ such that $\abdsetright$ is a collection of full 
\sttgds\; and $\sol{I}{\abdset}{abd}=\sol{I}{\Sigma}{cwa}=\sol{I}{\Sigma}{gcwa^*}$.
\end{corollary}  
\begin{proof}
For each atom $R(\vect{x})$ in the head of a full \tgd\; $\xi\in \Sigma$
add the following \abd\; to $\abdset$: 
$\body{\xi} \leftrightarrow R^i(\vect{x})$,
where $i$ is a new  annotation not used before.
It is easy to observe that for this $\abdset$ we have that 
$\abdsetright$ contains only full \tgds
\end{proof}

\setcounter{theorem}{7}
\begin{theorem}
Let $\abdset$ be a set of \abds\; and safe \aegds\;
with  $\ad{\abdset}=1$. 
Then for any $\ucqino$ query $\query{q}$
 the 
{\sc Eval}$_{\sf ABD}(\smap{\abdset}$,$\query{q})$ problem
is tractable and one may use only the universal representative, 
if it exists, 
to evaluate the query. 
\end{theorem}
\begin{proof}
Next we describe the polynomial algorithm that decides the evaluation problem
for mapping $\abdset$, source instance $I$ 
and boolean query $\query{q} \in \ucqino$.
Let $(T,\varphi^*)$ be a universal representative for $I$ and $\Sigma$. 
Consider query 
$\query{q}$  of the form:
$$\query{q}\;\defined\; \exists \vect{x}\; q(\vect{x}) \bigvee_{i\leq k} (q_i(\vect{x}_i)
\wedge x_i\neq y_i).$$
\noindent 
Where $q(\vect{x})$ is a $\ucq$ and each $q_i(\vect{x}_i)$ is a conjunctive 
query.
First, if $q(T)\hspace{-0,1cm}\downarrow$ 
(na\"ive evaluation of query $q$ over $T$) returns
\true, then the algorithm will return \true\; 
for $\query{q}$ too.
Using similar methodology as in the proof of Theorem 5.12
 \cite{DBLP:journals/tcs/FaginKMP05}, 
in case $q(T)\hspace{-0,1cm}\downarrow=\false$
we
construct the following set of \egds:
\begin{equation*}\label{eq1:eg:bissdirect5}
\begin{split}
q_1(\vect{x}_1) \rightarrow x_1 = y_1; \\
q_2(\vect{x}_1) \rightarrow x_2 = y_2; \\
\ldots \\
q_k(\vect{x}_1) \rightarrow x_k = y_k. 
\end{split}
\end{equation*}
\noindent 
Next we use the standard chase on $T$ with the previous set of \egds\;
also with respect to $\varphi^*$\footnote{This condition makes the difference 
between the certain answer result in OWA compared with the ABD-semantics.}.
Compared with the proof of Theorem 5.12 in 
 \cite{DBLP:journals/tcs/FaginKMP05}, because the \aegds\; in $\abdset$
are safe we can directly use the materialized target table $T$
and don't need to repeat the chase process 
for each query as done in  \cite{DBLP:journals/tcs/FaginKMP05}.
In case the chase process does fail
 either due
to trying to equate two distinct constants or due to 
contradicting with $\varphi^*$, then it returns \false. Otherwise it 
returns \true
\end{proof}

\bigskip
\setcounter{theorem}{8}
\begin{theorem}
There exists a  set
$\abdset$ 
 of \abds\; with $\ad{\abdset}=1$
and there is a conjunctive query with
two unequalities such that 
{\sc Eval}$_{\sf ABD}(\smap{\abdset}$,$\query{q})$ problem
is $\conp$-complete. 
\end{theorem}
\begin{proof}
This reduction from the 3CNF satisfiability problem
 is similar to the one provided by Madry in 
\cite{DBLP:journals/ipl/Madry05} with the set of $\abdset$ changed to
(note the new relation symbols $M$ and $V$):
\begin{equation*}\label{eq1:eg:bsidssirect5}
\begin{split}
& P(x,y) \leftrightarrow \annt{P'}{1}(x,y); \\
& L(x,y) \leftrightarrow \annt{P'}{2}(x,z),\annt{P'}{3}(z,y); \\
& R(x,y) \leftrightarrow \annt{P'}{4}(x,u),\annt{P'}{5}(u,v),\annt{P'}{6}(v,t),\\
& \hspace{4cm}\annt{P'}{7}(t,y), 
\annt{V'}{1}(x),\annt{M'}{1}(y);\\
& N(x,y) \leftrightarrow \annt{N'}{1}(x,y);\\
& V(x) \leftrightarrow \annt{V'}{2}(x); \\
& M(x) \leftrightarrow \annt{M'}{2}(x).
\end{split}
\end{equation*}
Where instance $I$ is defined similarly as in  \cite{DBLP:journals/ipl/Madry05} 
to which we add
$M^I=\{ (x_k^i) \setdef\; i \in \{1, \ldots,m \}, k \in \{ 1,2,3 \} \}$
and 
$V^I= \{ (T) \}$. The query is the same as the one in \cite{DBLP:journals/ipl/Madry05}. The new relational symbols $M$
and $V$ are needed in order to ensure that the third dependency
creates only paths between the variables from the 3CNF formula 
and vertex $T$
\end{proof}

\medskip
Some new definitions are needed in order to sketch the next proofs.
Some of these definitions are slight modifications of 
definitions from \cite{Onet12}. These modifications were needed
because of the different behavior between open and close nulls.

\medskip
\begin{definition}
Let $T$ be a semi-na\"ive table and $C$ be a finite set of 
constants. A $C$-retraction for $T$ is a mapping
$h$ from $\dom{T}$ to $\dom{T}\cup C$
that is identity on $\cons \cup \dom{h(T)}$.
\end{definition}

\medskip
\begin{definition}
Let $V$ be a na\"ive table and $T$ 
a semi-na\"ive table. 
A unifier for $V$ and $T$,
if it exists, is a pair 
$(\theta_1, (\theta_2,\{ \theta^1_2,\theta^2_2,\ldots,\theta^k_2 \}))$,
where $\theta_1$ is a homomorphism from 
the set $\dom{V}$ to $\dom{T} \cup \domc{V}$,
$\theta_2$ is a $\domc{V}$-retraction for $T$ 
identity on $\nullso$ 
and 
$\theta^i_2$, $1 \leq i \leq k$, 
is a $\domc{V}$-retraction for table $T$ 
identity on $\nullsc$ 
such that
$\theta_1(V)=\theta_2(\bigcup_{1\leq i \leq k } \theta^i_2(T))$.
Where $\domc{V}$ represents the set of constants occurring in $V$.
\end{definition}

\noindent
Note the asymmetrical role of instances $V$ and $T$.
In our case $V$ will usually play the role of the 
na\"ive table associated with a conjunctive query $\alpha$,
where each variable is replaced with a null value.
The semi-na\"ive table $T$ will be a subset of the 
universal representative $(T,\varphi^*)$ returned by the
annotated chase algorithm.

\medskip
\begin{definition}
A unifier 
$(\theta_1, (\theta_2,\{ \theta^1_2,\theta^2_2,\ldots,\theta^k_2 \}))$
 for tables 
$V$ and $T$ is {\em more general} than another unifier 
$(\mu_1, (\mu_2,\{ \mu^1_2,\mu^2_2,\ldots,\mu^k_2 \}))$,
if there exist  mappings 
$(f,\{ f^1,f^2,\ldots,f^k \})$ 
on $\dom{T}$, with at least one mapping not identity,
such that 
$\mu_2=f \circ \theta_2$ and $\mu_2^i=f^i \circ \theta^i_2$
for all $1 \leq i \leq k$.
\end{definition}

\medskip
\begin{definition}
A table $T$ is said to be {\em minimally unifiable} with 
table $V$ if for any unifier $\theta$ of $V$ and $T$,
there is no $T' \subsetneq T$ such that 
$\theta$ is unifier of $V$ and $T'$.
\end{definition}

From the previous definition we have the following important result
that ensures tractability for our last theorems.

\medskip
\begin{proposition}
If $T$ is minimally unifiable with 
table $V$, then $|T|\leq |V|$.
\end{proposition}

\medskip
\begin{definition}
A unifier 
$(\theta_1, (\theta_2,\{ \theta^1_2,\theta^2_2,\ldots,\theta^k_2 \}))$
 is a most general unifier (mgu)
for tables $V$ and $T$, if all unifiers 
$(\mu_1, (\mu_2,\{ \mu^1_2,\mu^2_2,\ldots,\mu^k_2 \}))$
 of $V$ and $T$ that are more general 
than unifier
$(\theta_1, (\theta_2,\{ \theta^1_2,\theta^2_2,\ldots,\theta^k_2 \}))$
 actually are isomorphic with it.
We denote $\mgu(V,T)$ the set of (representatives of the)
equivalence classes of all mgu's of $V$ and $T$.
\end{definition}

\medskip
By abusing the notations  with 
$\mgu(V,T)$ we will also denote the following set
$\mgu(V,T)\defined \bigcup_{T'\subseteq T} \mgu(V,T')$.
The set will contain only the representatives of the equivalence classes
of mgu's.
The following lemma ensures that the set 
of most general unifiers can be computed in polynomial time.

\medskip
\begin{lemma}\label{mgu}
Let $V$ be a na\"ive table and $T$ a semi-na\"ive table with
$c=|V|$ and $n=|T|$. Then one may compute the set
$\mgu(V,T)$ in $O((2c)^{2c}c^2n^{c!})$
\end{lemma}

\begin{proof}
The proof directly follows  from Proposition~18 in \cite{Onet12}
\end{proof}

\medskip
\setcounter{theorem}{9}
\begin{theorem}
Let $(T,\varphi^{*})$ be a universal representative
for some source instance $I$ and a set $\abdset$, 
 of \abds\; and safe \aegds\;with $\ad{\abdset}=1$.
Then there exists a polynomial time algorithm, 
 with input $(T,\varphi^*)$ and $\vect{t} \subset \cons$,
such that for any universal query $\query{q}$
decides if $\vect{t} \in \bigcap_{J \in \repe{T,\varphi^*}} \query{q}(J)$. 
\end{theorem}
\begin{proof}

A universal query is a $\fo$ query of the form 
$\query{q}(\vect{x})\; \defined\; \forall \vect{y}\; \varphi(\vect{x},\vect{y})$,
where $\varphi$ is a quantifier-free $\fo$ formula.
It is easy to verify that $\vect{t}\notin \cert{q}{I}{\abdset}{abd}$ if and only if
there exists a nonempty ground instance $J$ in $\sol{I}{\abdset}{abd}$
such that $J \models \neg \query{q}$. Note that 
$\neg \query{q}$ is logically equivalent with 
\cite{DBLP:journals/corr/abs-1107-1456}:

\begin{equation}\label{universal_not}
\begin{split}
\hspace{0cm}
\bar{\query{q}}(\vect{x}) \defined
\bigvee_{i\letbe 1}^{m} (\exists y_i \bigwedge_{j\letbe 1}^{n} (\delta_{ij}))\comma
\end{split}
\end{equation}
\noindent 
where $\delta_{ij}$ is an atomic formula or the negation of an 
atomic formula.
In order to compute $\eval{q}{\abdset}{abd}$ it is enough
to check that there exists an instance 
$J$ in $\sol{I}{\abdset}{abd}$ such that
$J \models \bar{\query{q}}(\vect{t})$. If such instance
exists, then $\eval{q}{\abdset}{abd}$ will return \false.
For this it is enough to find an integer $i \leq m$ and instance $J$
such that $J \models \varphi_i(\vect{t})$,
where: 
\begin{equation}\label{universal_not_and}
\begin{split}
\hspace{0cm}
\varphi_i(\vect{x}) \defined
\exists y_i \bigwedge_{j\letbe 1}^{n} (\delta_{ij}).
\end{split}
\end{equation}

Figure 1 lists the algorithm which for a fix
query of the form (\ref{universal_not_and}), a given
semi-na\"ive table $T$ 
and a tuple $\vect{t}$ 
decides
if there exists a ground instance $J \in \repe{T}$
with $J \models \varphi_i(\vect{t})$. Thus the 
result of $\eval{q}{\abdset}{abd}$ will be 
the negation of the result returned by the  
algorithm below. In the following we will also view  
$f=\{ x_1/a_1, x_2/a_2, \ldots, x_n/a_n \}$ 
as formula $(x_1=a_1 \wedge x_2=a_2 \wedge \ldots \wedge x_n=a_n)$.
Thus $\neg f$ will represent formula
$(x_1\neq a_1 \vee x_2\neq a_2 \vee \ldots \vee x_n\neq a_n)$.
A conjunctive formula $\bar{\query{q}}^+ \defined \bigwedge_{i\letbe 1}^{\ell} R_i(v(\vect{u}_i))$ we view it also as a na\"ive table 
where each variable from the formula is replaced by a new distinct null.

%%FIG 1 WAS HERE

For our problem membership problem
is $\vect{t} \in \bigcap_{J \in \repe{T,\varphi^*}} \query{q}(J)$?
we will return \true\; if
EXISTS\_EVAL($(T,\varphi^*)$,$\vect{t}$) 
returns \false, and it returns \false\; otherwise.
Note that the problem from Step 23 of the algorithm is solvable in 
polynomial time as $\varphi$ is a conjunction of disjunctive formulae 
of the following form $(x_1\neq y_1 \vee x_2\neq y_2 \vee \ldots \vee x_n\neq y_n)$,
 $p \leq |\bar{\query{q}}^+|$ 
and $\theta_1$, $\theta_2$ and each
$\theta^i_2$, $1\leq i \leq p$, are of fixed size bounded by
the maximum arity of a relation in $\bar{\query{q}}^+$
\end{proof}

\medskip
\setcounter{theorem}{10}
\begin{theorem}
Let $\abdset$ be a set of \abds\; and safe \aegds\; with $\ad{\abdset}=1$
and such that $\abdsetright$ is GAV-reducible and 
each \aegd\; does not equate two 
variables both occurring in affected positions.
Then for any $\cqnp$ query
the {\sc Eval}$_{\sf ABD}(\abdset,Q)$
problem is polynomial and  can be decided using 
a universal representative.
\end{theorem}
\begin{proof}
Let $\query{q}$ be a $\cqnp$, that is $\query{q}$ has exactly one
positive atoms. 
The polynomial algorithm listed in 
Figure 2  decides if 
$\vect{t} \in \cert{q}{I}{\abdset}{abd}$, 
where $I$ is an input instance and $T$ is the universal representative for
$I$ and $\abdset$. 
Table $T$ may be computed in polynomial time 
by the annotated-chase algorithm.
Because  
none of the \aegds\;  equate variable occurring in affected positions 
in the body of the \aegds, 
the same annotated chase algorithm
will return 
$T$ as a table containing only nulls from $\nullo$.

%FIG 2 WAS HERE

Few clarifications are in order here. 
Let $k$ be the maximum number of atoms occurring in the head of any
\tgd\; from $\abdsetright$.
It can be easily verified that for each block $B$
from Step 11 we have $|B|\leq k$. From this it follows that
the size of table $U_i$ from Step 13 is also bounded by the same
$k$.
At Step 14 for each $i$ the list of mgu's can be listed in time
$O((2k)^{2k}k^2n^{k!})$, where $n=|T|$. Note that in this case
the mgu considers two semi-na\"ive tables that may share null values.
At Step 17, the size of instance $J_i$ is not bounded by any constant,
still we may verify if for an instance $J_i$ and table $T$ if there 
exists the set of homomorphisms,
as per Step 13, 
with the following algorithm:

\begin{codebox}
\Procname{STEP18\_CHECK($T$,$J_i$)}
\li Let $v: J_i \rightarrow \{ \true,\false \}$;
\li Let $v(\vect{t})\defined \false$, for all $\vect{t} \in J_i$;
\li \For all $\vect{s} \in T$ and all $h$ such that $h(\vect{s})\in J_i$
\li \Do 
\li Let $v(h(\vect{s}))\defined \true$; \End 
\li \If $v(\vect{t})=\true$ for all $\vect{t}\in J_i$ 
\li \Then \Return \true;
\li \Else \Return \false; \End
\end{codebox}

It can be verified that $\cqn$\_EVAL($T$,$\vect{t}$)
algorithm is sound and complete in deciding
if $\vect{t} \in \cert{q}{I}{\abdset}{abd}$
\end{proof}

\begin{figure*}[h]\label{alg1}
\caption{EXISTS\_EVAL algorithm}
\begin{codebox}
\Procname{EXISTS\_EVAL($(T,\varphi^*)$,$\vect{t}$)}
\li Let $\bar{\query{q}}(\vect{x}) \defined 
\exists \vect{y} \bigwedge_{i\letbe 1}^{\ell} R_i(\vect{u}_i) \wedge 
\bigwedge_{j\letbe 1}^{m} \neg S_j(\vect{w}_j) 
\bigwedge_{k\letbe 1}^{n} \neg t_{k1} = t_{k2}$;
\li \ \ \ \ \ \ \ \ \ \ \ \ \ \ 
where $\vect{x} \subseteq \big( \bigcup^{\ell}_{i\letbe 1} \vect{u}_i \cup \bigcup^m_{j\letbe 1} \vect{w}_j\big )$;
\li \ \ \ \ \ \ \ \ \ \ \ \ \ \ \ \ \ \ \ \ \ \ \ 
 $\vect{y} = (\bigcup^{\ell}_{i\letbe 1} \vect{u}_i \cup \bigcup^m_{j\letbe 1} \vect{w}_j) \setminus \vect{x}$; and
 \li \ \ \ \ \ \ \ \ \ \ \ \ \ \ \ \ \ \ \ \ \ \ \ 
 $t_{k1}, t_{k2} \in  \vect{x} \cup \vect{y} \cup \cons$, for $1 \leq k \leq n$;
\li
\li Let $v : \vect{x} \cup \vect{y} \rightarrow \vect{t} \cup \vect{y}$, such that $v(\vect{x})=\vect{t}$ and $v(y)=y$ for all $y \in \vect{y}$;
\li
\li Let $\bar{\query{q}}^+ \defined \bigwedge_{i\letbe 1}^{\ell} R_i(v(\vect{u}_i))$;
\li
\li \For all 
$(\theta_1, (\theta_2,\{ \theta^1_2,\theta^2_2,\ldots,\theta^p_2 \}))$ 
mgu for $\bar{\query{q}}^+$ and $T$ 
\li \Do
\li  Let $\varphi \letbe \varphi^*$;
\li \For j:=1 \To m 
\li \Do 
\li \For all $(\xi^j_1,(\xi^j_2,\{ \bar{\xi}^j_2 \}))$ mgu for $S_j(v(\vect{w}_j))$ and $T$ 
\li \Do
\li  Let $\varphi \letbe \varphi \wedge (\neg \xi^j_1 \vee \neg \xi^j_2 \vee \neg \bar{\xi}^j_2)$;
\li \End 
\li \For k:=1 \To n 
\li \Do 
\li Let $\varphi \letbe \varphi \wedge  (v(t_{k_1})\neq v(t_{k_2}))$;
\li  \End \End
\li \If $(\theta_1 \wedge \theta_2 \wedge \theta^1_2 \wedge \theta^2_2 \wedge \ldots \wedge \theta^p_2 \wedge \varphi)$ is satisfiable
\li \Then \Return $\true$
\li \End \End
\li  \Return $\false$
\end{codebox}
\end{figure*}

\begin{figure*}\label{alg2}
\caption{$\cqn$\_EVAL algorithm}
\begin{codebox}
\Procname{$\cqn$\_EVAL($T$,$\vect{t}$)}
\li Let $\query{q}(\vect{x}) \defined 
\exists \vect{y}\; 
R_i(\vect{u}) \wedge \bigwedge_{i\letbe 1}^{n} \neg S_{i}(\vect{w}_{i})$;
\li \ \ \ \ \ \ \ \ \ \ \ \ \ \ \ \ \ \ \ where $\vect{x}\cup \vect{y} = \vect{u}\setminus \cons$; and
\li \ \ \ \ \ \ \ \ \ \ \ \ \ \ \ \ \ \ \ \ \ \ \ \ \ \ \ \  $\bigcup^{n}_{i\letbe 1} \vect{w}_i \setminus \cons \subseteq \vect{u}\setminus \cons$;
\li Let $v : \vect{x} \cup \vect{y} \rightarrow \vect{t} \cup \vect{y}$, such that $v(\vect{x})=\vect{t}$ and $v(y)=y$ for all $y \in \vect{y}$;
\li Let $g_{i}$ be the function that maps vector 
$\vect{u}$ to $\vect{w}_{i}$, 
for $i\in \{1,\ldots,n \}$;
\li 
\li Let $V\defined\{ R(h\circ v(u))\setdef R(h\circ v(u))\in T \mbox{ for homomorphism } h \}$;
\li Let $J=V \cap (\cons)^{\ar{R}}$;
\li Let $V'=V \setminus J$;
\li
\li \For each set $B$ from the Gaifman-partition of $V'$ 
\li \Do
\li Let $U_{i}\defined\{ S_{i}(g_{i}(\vect{z}))\setdef R(\vect{z})\in B \}$,
for $i\in \{1,\ldots,n \}$;
\li \If $\neg$($\exists i$ and 
$(\theta_{1,i},(\theta_{2,i},\{\theta_{2,i}^1, \theta_{2,i}^{2},\ldots,\theta_{2,i}^{|B|} \}))$ mgu for $U_{i}$ and $T$)
\li \Then \Return \true; \End \End 
\li 
\li Let $J_{i}\defined\{ S_{i}(g_{i}(\vect{z}))\setdef R(\vect{z})\in J \}$,
for $i\in \{1,\ldots,n \}$;
\li \If $\neg$($\exists i$, $\{h,h_1,h_2,\ldots,h_n \}$ 
 and $T'\subset T$ such that $h(\bigcup_{i\letbe 1}^n h_i(T'))=J_i$ )
\li \Then \Return \true; 
\li \End
\li  \Return \false;
\end{codebox}
\end{figure*}

\end{document}